\documentclass[final,onefignum,onetabnum]{siamart190516}

\usepackage{amsmath}
\usepackage{amssymb}
\usepackage{bbm}
\usepackage{cleveref}
\usepackage{tabularx}
\usepackage{multirow}
\usepackage{hhline}
\usepackage{algcompatible}
\usepackage{algorithm}
\usepackage{enumitem}
\usepackage{multicol}
\usepackage{bm}
\usepackage{svg}
\usepackage{caption}
\usepackage{subcaption}
\usepackage{titlesec}

\titlespacing\section{0pt}{0pt plus 0pt minus 1pt}{0pt plus 0pt minus 1pt}
\titlespacing\subsection{0pt}{0pt plus 0pt minus 1pt}{0pt plus 0pt minus 1pt}
\titlespacing\subsubsection{0pt}{0pt plus 0pt minus 1pt}{0pt plus 0pt minus 1pt}

\newcommand{\mcl}{\mathcal}

\newcommand{\bz}{\mathbf{z}}

\newcommand{\mbb}{\mathbb}
\newcommand{\lp}{\left(}
\newcommand{\rp}{\right)}

\newtheorem{assumption}{Assumption}

\def\eps{{\epsilon}}

\def\vv{{\bm{v}}}
\def\vw{{\bm{w}}}
\def\vx{{\bm{x}}}

\def\mA{{\bm{A}}}

\def\mD{{\bm{D}}}

\def\mF{{\bm{F}}}
\def\mG{{\bm{G}}}

\def\mI{{\bm{I}}}

\def\mL{{\bm{L}}}
\def\mM{{\bm{M}}}

\def\mP{{\bm{P}}}

\def\mX{{\bm{X}}}
\def\mY{{\bm{Y}}}

\def \RR {{\mathbb{R}}}
\def \EE {{\mathbb{E}}}

\def \bX {{\bm X}}

\usepackage{verbatim}

\usepackage{wrapfig}

\headers{Efficient and Reliable Overlay Networks 
for DFL}{Y. Hua, K. Miller, A. Bertozzi, C. Qian, and B. Wang}

\title{Efficient and Reliable Overlay Networks 
for Decentralized Federated Learning\thanks{Submitted to the editors DATE.}}

\author{Yifan Hua\thanks{Department of Computer Science and Engineering, University of California, Santa Cruz
  (\email{yhua294@ucsc.edu}) (Co-first).}
\and Kevin Miller\thanks{Department of Mathematics, University of California, Los Angeles
  (\email{millerk22@math.ucla.edu}) (Co-first).}
\and Andrea L. Bertozzi\thanks{Department of Mathematics, University of California, Los Angeles
  (\email{bertozzi@math.ucla.edu}) (Co-last).}
 \and Chen Qian\thanks{Department of Computer Science and Engineering, University of California, Santa Cruz
  (\email{cqian12@ucsc.edu}) (Co-last).}
\and Bao Wang\thanks{Department of Mathematics, Scientific Computing and Imaging Institute, University of Utah
  (\email{wangbaonj@gmail.com}) (Co-last).}
  }

\ifpdf
\hypersetup{
  pdftitle={Efficient and Reliable Overlay Networks for Decentralized Federated Learning},
  pdfauthor={Y. Hua, K. Miller, A. L. Bertozzi, C. Qian, and B. Wang}
}
\fi

\begin{document}

\maketitle

\begin{abstract}
We propose near-optimal overlay networks based on $d$-regular expander graphs to accelerate decentralized federated learning (DFL) and improve its generalization. In DFL a massive number of clients are connected by an overlay network, and they solve machine learning problems collaboratively without sharing raw data. Our overlay network design integrates spectral graph theory and the theoretical convergence and generalization bounds for DFL. As such, our proposed overlay networks accelerate convergence, improve generalization, and enhance robustness to clients failures in DFL with theoretical guarantees. Also, we present an efficient algorithm to convert a given graph to a practical overlay network and maintaining the network topology after potential client failures. We numerically verify the advantages of DFL with our proposed networks on various benchmark tasks, ranging from image classification to language modeling using hundreds of clients.
\end{abstract}

\begin{keywords}
Decentralized federated learning; Overlay networks; Random graphs.
\end{keywords}

\begin{AMS}
 65B99, 68T01, 68T09, 68W15.
\end{AMS}

\section{Introduction}~
{\em Federated Learning} (FL) is a machine learning (ML) setting where a massive number of entities (clients) solve an ML problem collaboratively without transferring raw data, under the coordination of a central server \cite{pmlr-v54-mcmahan17a,kairouz2019advances}. FL trains ML models by exchanging the model parameters between clients and the central server; in each communication round, the central server distributes parameters to clients and aggregates the updated parameters from clients. FL decouples the model training from the need for collecting or direct access to the private training data; therefore, FL significantly reduces privacy and security risks. Many algorithms have been developed for 
FL, such as FedAvg~\cite{pmlr-v54-mcmahan17a}, SCAFFOLD~\cite{karimireddy2020scaffold}, FedProx~\cite{LiSZSTS20}, FedPD~\cite{zhang2020fedpd}, FedSplit \cite{pathak2020fedsplit}, and FedOpt~\cite{reddi2020adaptive}. Compared to many distributed optimization settings \cite{nedic2009distributed,mcdonald2010distributed,balcan2012distributed,zhang2014deep,povey2014parallel,fercoq2014fast,shamir2014distributed}, FL gains tremendous advantages in communication efficiency. We can mathematically formulate FL as solving the following optimization problem
\begin{equation}\label{eq:FL:Opt}
{\small \min_{{\vw}\in \RR^d} f({\vw}):=\frac{1}{N}\sum_{i=1}^Nf_i({\vw}),\
}
\end{equation}
where $f_i({\vw})=\EE_{({\vx}, y)\sim \mathcal{D}_i}\mathcal{L}(g({\vx}, {\vw}), y)$ with 
$({\vx},y)$ be a data-label pair sampled from the data distribution $\mathcal{D}_i$ on 
the $i^{th}$ client, and $g(\cdot,\vw)$ is the ML model. 
As shown in Fig.~\ref{fig:graphs} (a), in the $i^{th}$ communication round, FedAvg \cite{pmlr-v54-mcmahan17a}, one of the most popular FL algorithms, iterates as follows: the server (node 1) sends the current parameters ${\vw}_i$ to a small fraction of selected clients $\{k_j|k_j\in \{1,2,\cdots,N\}, \mbox{for}\ j=1,2,\cdots,m\}$. Each selected client then updates ${\vw}_i$ for $T$ iterations by using its local data and stochastic gradient-based algorithms.
The 
server then 
aggregates these locally updated parameters to get the 
updated model 
after the current communication round.
The existence of central server raises several 
concerns about FL: 1) the communication cost between the 
server and clients can be excessive since a large number of clients are 
involved in a practical FL system, 2) the failure of the server would disrupt the training process of all clients, and 3) the privacy of the whole FL system can be fragile since the central server is exposed to adversaries.

\begin{figure}[!ht]
\centering
\begin{tabular}{cccc}
\includegraphics[clip, trim=3.0cm 6.0cm 3.0cm 7.0cm, width=0.2\columnwidth]{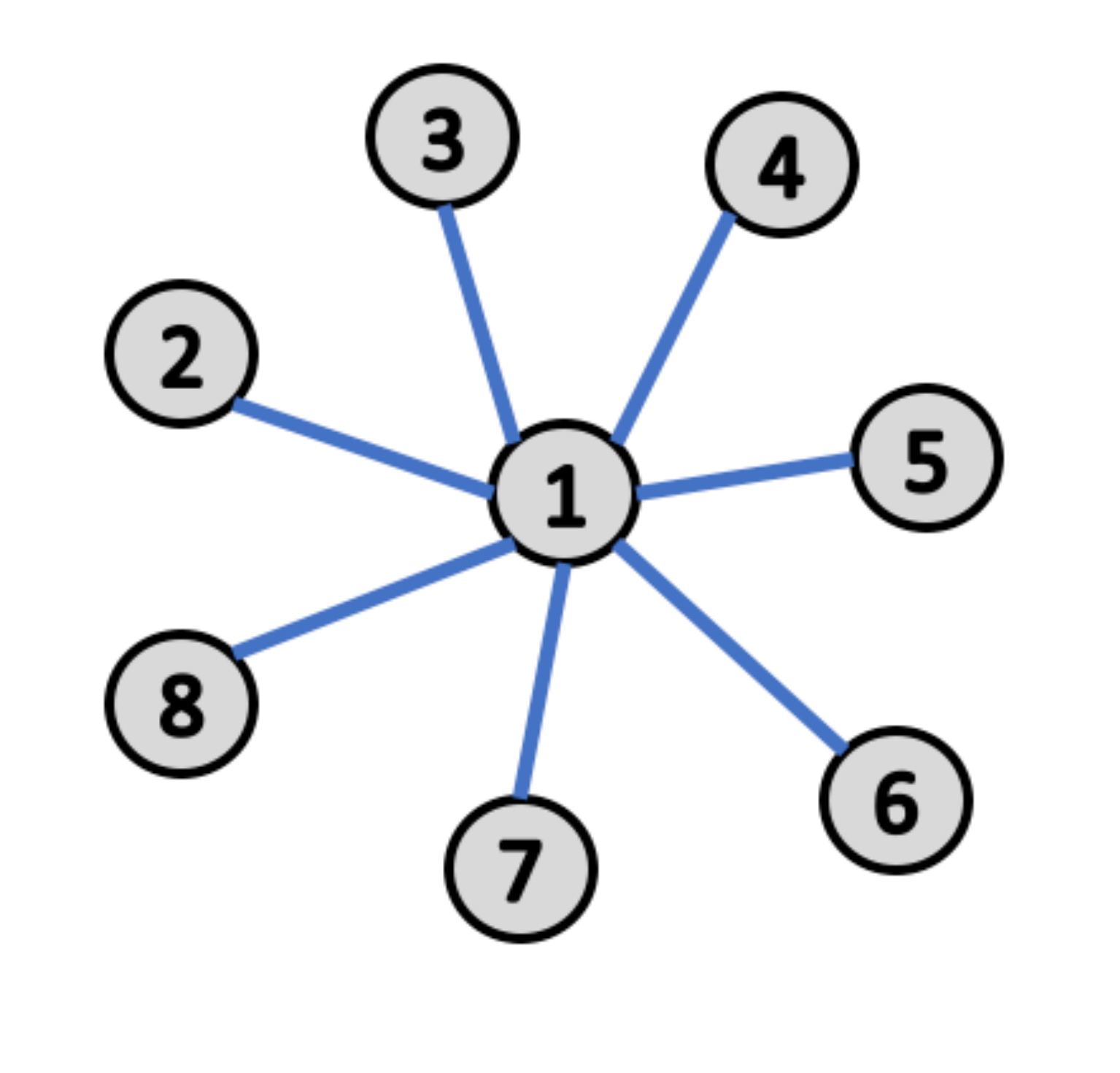}&
\includegraphics[clip, trim=3.0cm 6.0cm 3.0cm 7.0cm, width=0.2\columnwidth]{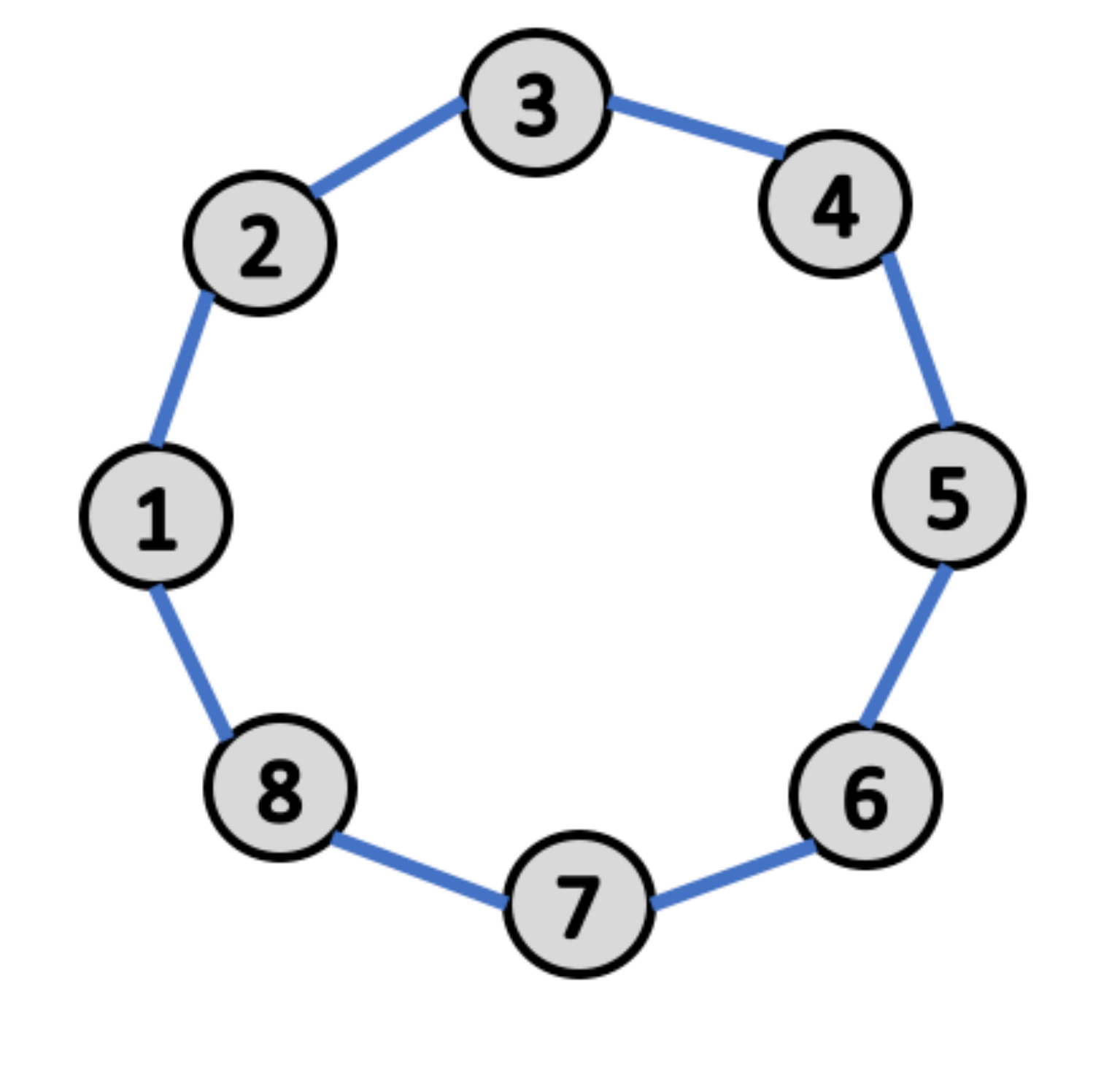}&
\includegraphics[clip, trim=3.0cm 6.0cm 3.0cm 7.0cm, width=0.2\columnwidth]{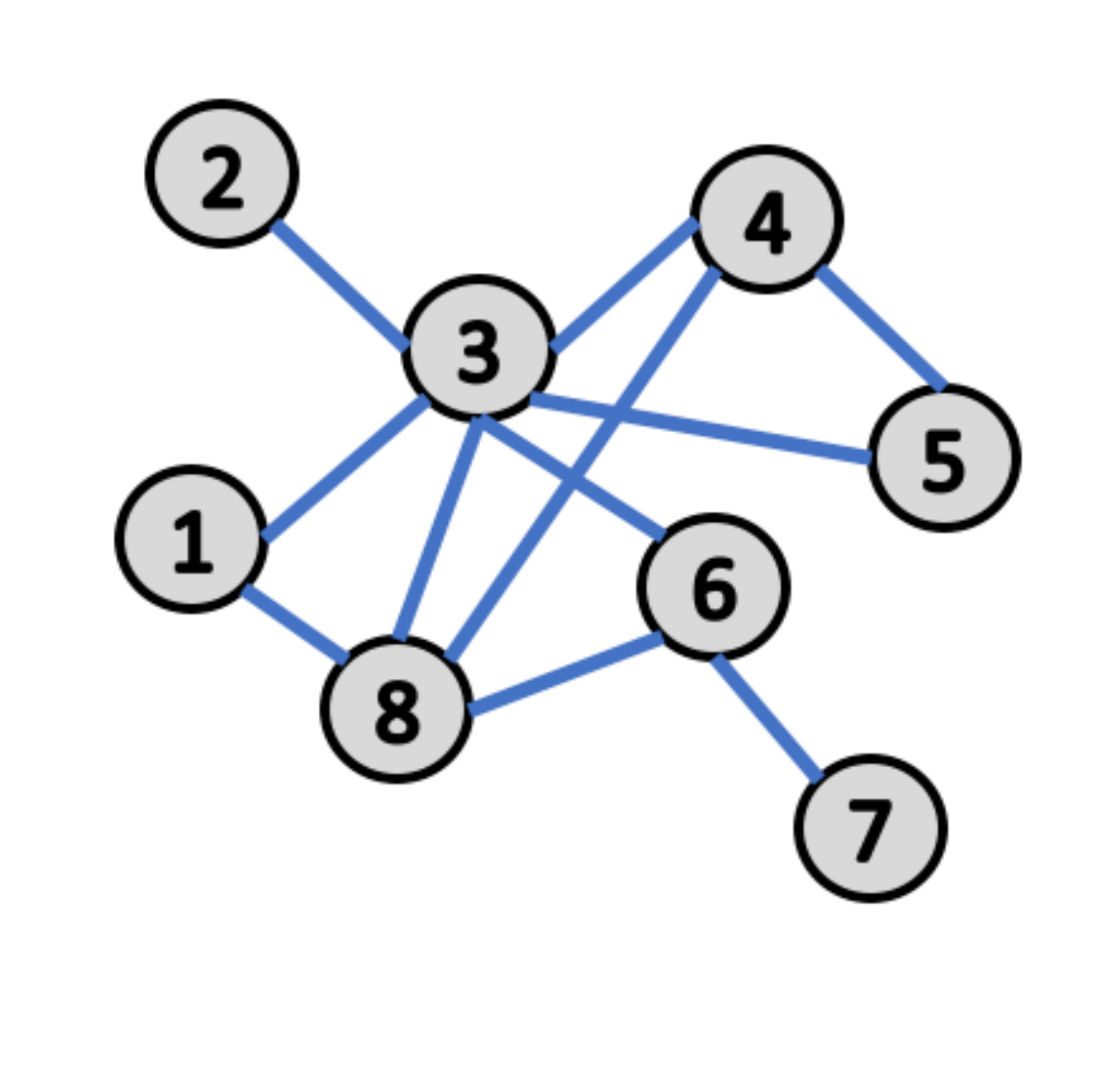}&
\includegraphics[clip, trim=3.0cm 6.0cm 3.0cm 7.0cm, width=0.2\columnwidth]{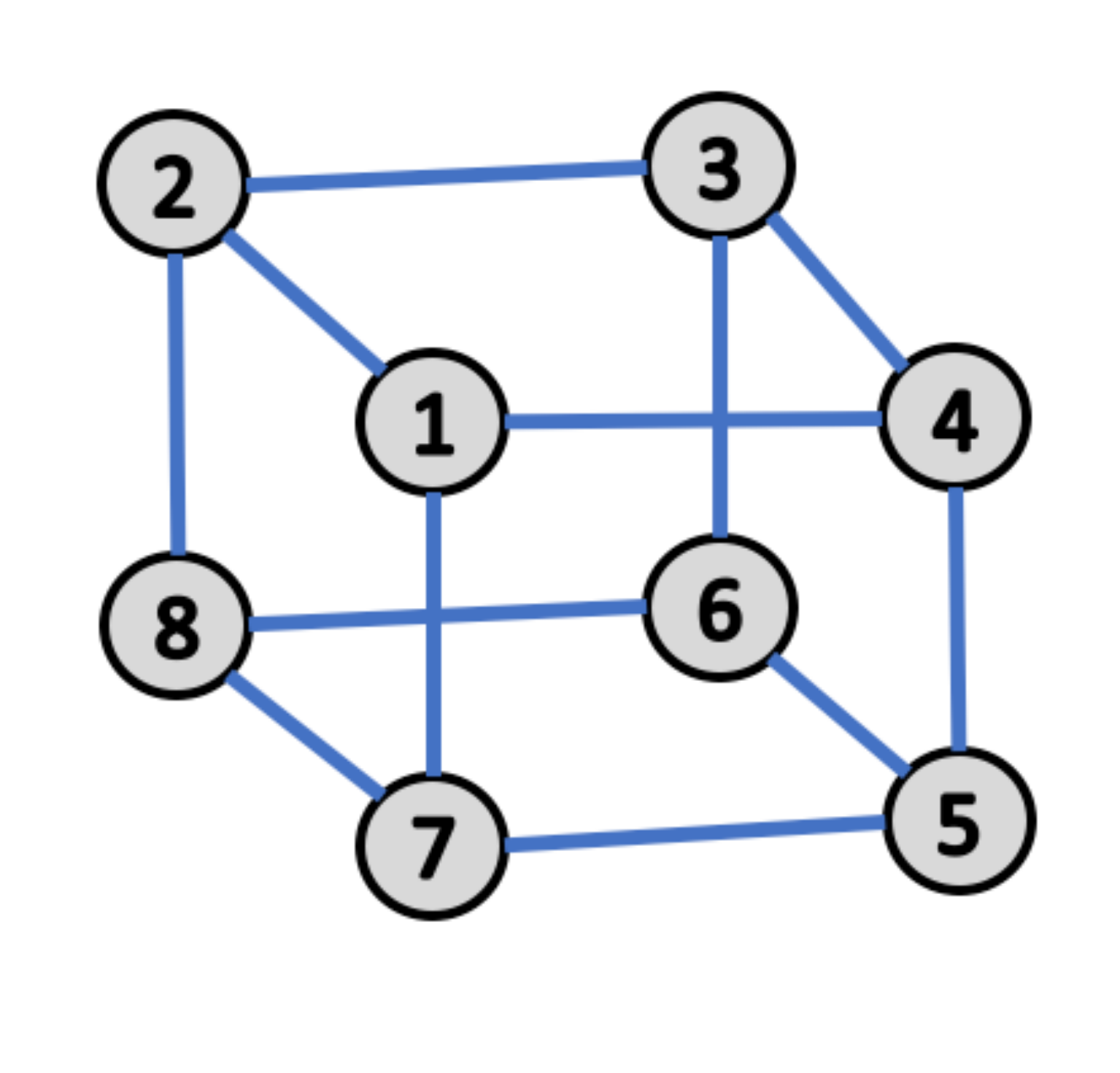}\\[-8pt]
  {\small (a) FL} 
& {\small (b) Ring} 
& {\small (c) Erd\"os-R\'enyi} 
& {\small (d) Expander} 
\\
\end{tabular}
\vspace{-0.3cm}
\caption{
Illustration of the network topology for federated learning and decentralized federated learning with Ring, Erd\"os-R\'enyi, and expander graphs.
}
\label{fig:graphs}
\end{figure}

\emph{Decentralized Federated Learning} (DFL) replaces the server-clients communication with client-client (peer-to-peer) communication, which significantly reduces the communication burden and privacy risks \cite{bonawitz2016practical,geyer2017differentially,orekondy2018gradient,truex2019hybrid,triastcyn2019federated,xu2019hybridalpha,liu2019enhancing,melis2019exploiting,liang:2020LSFed,choudhury2020anonymizing,liu2020fedsel,truex2020ldp,wei2020federated,liu2020privacy,abramson2020distributed,sun2021decentralized}. In DFL, all clients are connected by an overlay network, e.g. Fig.~\ref{fig:graphs} (b) Ring, (c) Erd\"os-R\'enyi, and (d) $d$-regular expander graphs. 
The clients update in the same way as that in FL, and each client only sends its locally updated model to its topological neighbors and aggregates the updated models from its neighbors. 
Network topology has a profound impact on the convergence, generalization, and robustness of DFL. 
In this paper, \emph{
we focus on designing efficient network topologies that guarantee fast and accurate DFL, and resilient 
to client failures.}

\subsection{Our contribution}~
Based on the theoretical convergence rate \cite{yujin2019,sun2021stab,sun2021decentralized} and our first established generalization bound of DFL, where each client trains ML models using stochastic gradient descent with momentum, we design near optimal network topology to connect clients to train ML models collectively. In particular, leveraging the random graph theory, we propose $d$-regular expander graphs for the network topology, which is provably to be near-optimal. The major advantages of leveraging $d$-regular expander graphs for the overlay networks design are threefolds:

\begin{itemize}[leftmargin=*]
\item DFL with $d$-regular expander graphs converges remarkably faster and generalizes better than DFL using other sparse graphs, including Ring and Erd\"o-R\'enyi graphs. 


\item Expander graph connects each node with $d$ neighbors, resulting in low communication cost in the decentralized federated learning.


\item DFL with $d$-regular expander graphs enables robust decentralized federated learning 
with respect to potential node failures.

\end{itemize}

\subsection{Additional related works}~
\paragraph{Network Design.} Chow et al. \cite{chow2016expander} have designed expander graphs for decentralized optimization using deterministic local optimization algorithms. In \cite{marfoq2020throughput}, the authors use theory of the max-plus linear systems and design efficient topology for cross-silo FL, in which close-by data silos can exchange information faster with the central server. We focus on designing efficient networks for cross-device DFL that are scalable to a massive amount of devices. 

\paragraph{Analysis of DFL/FL algorithms.} 
The convergence properties of FedAvg or local SGD have been studied extensively \cite{wang2018cooperative, stich_error-feedback_2020, jiang_agrawal_2018}, mainly focusing on the IID case. Non-IID convergence for FL has been shown in \cite{jiang_agrawal_2018,zhang2020fedpd,yujin2019,li_convergence_2020}. Convergence analysis of DFL has been shown in \cite{wang2019matcha,sun2021decentralized}. While convergence analysis for the myriad of problem setups has been provided, generalization guarantees have been more elusive.

Convergence analysis of DFL hinges on connectedness properties of the underlying graph topology, captured in the spectral properties of the associated mixing matrix (see Section~\ref{sec:theory}).  The authors of \cite{wang2018cooperative} discuss how different versions of local SGD correspond to different graph topologies and \cite{wang2019matcha} provides an efficient decomposition of graph topology for improved communication costs.

\paragraph{Practical network construction.} 
Building overlay networks have been studied extensively in previous works. However, in the past, overlay networks are mainly used for peer-to-peer file sharing \cite{P2Psurvey}, online social networks \cite{OSNsurvey}, and routing infrastructures \cite{MDT,ROME-ToN}.
For peer-to-peer file-sharing networks, existing studies have proposed to utilize random walks to achieve distributed $d$-regular expander graphs with assuming each node could choose $d$ neighbor at random \cite{feder2006FOCS,Law2003INFOCOM}. However, such assumption does not hold in DFL because no node can uniformly choose $d$ neighbors among existing nodes at random since there is no central coordinator. 
However it is possible to build an expander graph with tight connectivity if the global information are given such as maintaining distributed Delaunay triangulation graphs for wireless sensor networks \cite{MDT}, metro Ethernet \cite{ROME-ToN}, random regular graphs for data center networks \cite{S2-TPDS}, and memory interconnection networks \cite{StringFigure}.

\subsection{Notations}~
We denote scalars by lower or upper case letters; vectors and matrices by lower and upper case boldface letters, respectively. For a vector $\vx = (x_1, \cdots, x_d)^\top\in \mathbb{R}^d$, 
we use $\|\vx\| := {(\sum_{i=1}^d |x_i|^2)^{1/2}}$ and $\|\vx\|_\infty := \max_{i=1}^d|x_i|$ to denote its $\ell_2$- and $\ell_\infty$-norm, respectively.
We denote the vector whose entries are all 0s as $\mathbf{0}$. For a matrix $\mA$, we use $\mA^\top$,  $\mA^{-1}$, and $\|\mA\|$ 
to denote its transpose, inverse, and spectral norm, respectively.
We denote the identity matrix as $\mI$.
For a function $f(\vx): \mathbb{R}^d \rightarrow \mathbb{R}$, we denote 
$\nabla f(\vx)$ as its gradient. 
Given two sequences $\{a_n\}$ and $\{b_n\}$, we write $a_n=\mathcal{O}(b_n)$ if there exists a positive constant $C$ 
such that $a_n \leq C b_n$.

\subsection{Organization}~
We organize this paper as follows: In Section~\ref{sec:theory}, we present the 
theoretical results for DFL 
on convergence rate and generalization bound. Based on these theoretical results 
we present our network topology design and its practical implementation 
in Sections \ref{sec:network} 
and \ref{sec:practical-network}, 
respectively. We verify the efficiency and robustness to the potential node failures of DFL with the designed network topology on various benchmarks 
in Section~\ref{sec:experiments}. 
Technical proofs 
are provided in the appendix.

\section{
Theory of DFedAvg}\label{sec:theory}~
An important notion in DFL 
is the \emph{mixing matrix}, which is associated with an undirected connected graph $\mathcal{G}=(\mathcal{V},\mathcal{E})$, with vertex set $\mathcal{V}=\{1,2,\cdots,N\}:=[N]$ and edge set $\mathcal{E}\subset \mathcal{V}\times \mathcal{V}$,
and the edge $(i,j)\in \mathcal{E}$ represents a communication channel between clients $i$ and $j$.
\begin{definition}[Mixing matrix]
\label{def:mixing-matrix}
A matrix ${\mM}=[m_{i,j}]\in \RR^{N\times N}$ is a mixing matrix, if it satisfies 1. (Graph) If $i\neq j$ and $(i,j)\notin \mathcal{E}$, then $m_{i,j}=0$, otherwise, $m_{i,j}>0$; 2. (symmetry) ${\mM}={\mM}^\top$; 3. (Null space property) $null\{{\mI}-{\mM}\}=span\{\mathbf{1}\}$ where ${\mI}\in \RR^{N\times N}$ and $\mathbf{1}\in \RR^{N}$ are the identity matrix and the vector whose entries are all $1$s; 4. (Spectral property) ${\mI}\succeq {\mM}\succ -{\mI}$, where ${\mI}\succeq {\mM}$ means ${\mI}-{\mM}$ is positive semi-definite and ${\mM}\succ -{\mI}$ stands for ${\mM}+{\mI}$ is positive definite. 
\end{definition}
Given the adjacency matrix of a network, its maximum-degree matrix and metropolis-hastings matrix are both mixing matrices \cite{boyd2004fastest}. The symmetric property of ${\mM}$ indicates that its eigenvalues are real and can be sorted in the non-increasing order. Let $\lambda_i({\mM})$ denote the $i^{th}$ largest eigenvalue of ${\mM}$, then we have  $\lambda_1({\mM})=1>\lambda_2({\mM})\geq \cdots \geq \lambda_N({\mM})>-1$ based on the spectral property of the mixing matrix.
The mixing matrix also serves as a probability transition matrix of a Markov chain. 
An important constant
is $\lambda=\lambda({\mM}):=\max\{|\lambda_2({\mM})|,|\lambda_N({\mM})|\}$, which describes the speed of the Markov chain, 
induced by the mixing matrix ${\mM}$, converges to its stable state. 


We consider DFL using the following update on client $i$ 
\begin{equation} \label{eq:DFedAvgM-local-update}
{\small     \vw^{t,k+1}_i  = \vw^{t, k}_i -\eta_t \nabla f_i(\vw^{t,k}_i; \xi^{t,k}_i) + \beta(\vw^{t,k}_i - \vw^{t,k-1}_i),}
\end{equation}
where $t$ is the communication round, $k$ is the local 
iteration, and $\xi^{t,k}_i = (\vx^{t,k}_i, y^{t,k}_i) \sim \mcl D_i$. 
After the $K^{th}$ local iteration, 
communication happens according to the graph topology of the mixing matrix, ${\mM}$; that is, we have for each  $i \in [N]$:
\[
    \vw^{t+1, 0}_i = \sum_{\ell =1}^N m_{i,\ell} \vw^{t,K}_\ell.
\]
To ensure well-defined iterations, we set $\vw^{t,-1}_i = \vw^{t,0}_i$ for each $i$. These iterations are referred to as DFedAvgM (Decentralized Federated Averaging with Momentum) \cite{sun2021decentralized}.

To guarantee convergence of generalization of DFedAvgM, we collect below the necessary assumptions on the local functions $f_i$ and global function $f$:
\begin{assumption}[L-smooth] \label{assumption:Lsmooth}
     $f_1,\ldots,f_m$ are all L-smooth, i.e. $f_i(\vw)\le f_i(\vv)+\langle\nabla f_i(\vv),\vw-\vv\rangle+\frac{L}{2}\|\vw-\vv\|_2^2$ for all $\vw, \vv$.
\end{assumption}
\begin{assumption}[Bounded Local Gradient Variance (BLGV)] \label{assumption:bdd-stoch-var}
     Let {\small$\xi^t_i := (\vx^{t, k}_i, y^{t,k}_i)$} be sampled from the $i^{th}$ device's local data $\mcl D_i$ uniformly at random. Then for all $i \in [N]$: 
     $\mbb E\|\nabla f_i(\vw^{t,k}_i; \xi^{t,k}_i) - \nabla f_i(\vw^{t,k}_i)\|_2^2 \le \sigma^2$, i.e. the stochastic gradients have bounded variance.
\end{assumption}
\begin{assumption}[Bounded Global Gradient Variance (BGGV)]\label{assumption:bdd-global-var}
    The global variance is bounded, i.e. $\| \nabla f_i(\vw) - \nabla f(\vw)\|^2 \le \zeta^2$.
\end{assumption}
\begin{assumption}[Bounded Local Gradient Norm (BLGN)]\label{assumption:bdd-grad}
    At each node $i \in \{1, \ldots, m\}$, the norm of the gradients is uniformly bounded, i.e. $\max_{\vw} \|\nabla f_i(\vw)\| \le B$.
\end{assumption}

While convergence guarantees for FL and DFL have been studied extensively~\cite{wang2018cooperative,stich_error-feedback_2020,jiang_agrawal_2018,sun2021decentralized}, we provide stability analysis for DFedAvgM to give generalization guarantees under Assumptions~\ref{assumption:Lsmooth}-\ref{assumption:bdd-grad}. Along with related convergence guarantees, our work here elucidates the importance of beneficial graph topology design. 

\subsection{Convergence 
of DFedAvgM}\label{sec:conv-result}~
We state convergence results for DFedAvgM and highlight the effect of graph topology on convergence rates in DFL. This result analyzes the convergence of 
the sequence $\{\bar{\vw}^t\}_{t=1}^T$ over the $T$ communication rounds, where $\bar{\vw}^t := \frac{1}{N} \sum_{i=1}^N \vw^t_i$ is the averaged weight vector over all the nodes. 
The following result comes from \cite{sun2021decentralized}.
\begin{theorem}[General nonconvexity \cite{sun2021decentralized}] \label{thm:gen-nonconvexity}
Let the sequence $\{\bar{\vw}_i^t\}_{t\ge 0}$ be generated by the DFedAvgM for each $i=1, 2, \ldots, N$, and suppose Assumptions~\ref{assumption:Lsmooth}-\ref{assumption:bdd-global-var} hold. Moreover, assume the constant stepsize $\eta$ satisfies $0< \eta \le 1/8LK$ and $64L^2K^2\eta^2 + 64LK\eta < 1$, where $L$ is the Lipschitz constant from Assumption~\ref{assumption:Lsmooth} and $K$ is the number of local updates before communication. Then,
{\small\begin{equation} \label{eq:conv-bound}
    \min_{1 \le t \le T}\ \mbb E\|\nabla f(\bar{\vw}^t)\|^2 \le \frac{2\bar{\vw}^1 - 2 \min f}{\gamma(K, \eta)T} + \alpha(K, \eta) + \frac{\Xi(K, \eta)}{(1-\lambda)^2},
\end{equation}}
where $T$ is the total number of communication rounds and $\gamma(K,\eta), \alpha(K, \eta)$, and $\Xi(K, \eta)$ are constants, and the detailed forms are given in the appendix.


\end{theorem}

In this result we clearly see the convergence of the auxiliary sequence depends on the value of $\lambda \in (0,1)$; namely, the closer that $\lambda$ is to 1, the worse the convergence bound of the final term of \eqref{eq:conv-bound}. In~\cite{yujin2019}, a similar dependence on this graph-dependent value $\lambda$ appears in their convergence result for a slightly different version of DFL with momentum. All this motivates selecting a graph topology that will minimize the value of $\lambda$.

\subsection{Generalization of DFedAvgM}~
In this section, we will establish a generalization bound of DFedAvgM. 
Given an algorithm $\mcl A$ that acts on data $\mcl D$ with output $\mcl A(\mcl D)$, the generalization error is given by {\small$\eps_{gen} :=\mbb E_{\mcl D, \mcl A} [ F(\mcl A(\mcl D)) - F_{\mcl D}(\mcl A(\mcl D))]$}, where {\small$F(\vx) = \mbb E_{\xi \sim \mcl D} f(\vx; \xi)$} is the {\it ``true'' risk} and {\small$F_{\mcl D}(\vx) = \sum_{i=1}^N f(\vx; \xi)/N$} is the {\it empirical risk} of the machine learning model for input $\vx$ with loss function $f$. Uniform stability is a useful property used to bound the generalization error $\eps_{gen}$, see e.g. \cite{Hardt2016, elisseeff2005stability}. 
\begin{definition}
A randomized algorithm $\mcl A$ is $\eps$-{\it uniformly stable} if for any two data sets $\mcl D, \mcl D'$ with $N$ samples each that differ in one example we have 
\[
    \sup_{\xi} \mbb E_{\mcl A} [ f(\mcl A(\mcl D); \xi) - f(\mcl A(\mcl D'); \xi)] \le \eps.
\]
\end{definition}
With this definition in hand, it has been proven that uniform stability implies bounded generalization error:
\begin{lemma}[
\cite{Hardt2016}] \label{lemma:gen-unif-stab}
Let $\mcl A$ be $\eps$-uniformly stable, then it follows that $$|\mbb E_{\mcl D, \mcl A}[ F(\mcl A(\mcl D)) - F_{\mcl D}(\mcl A(\mcl D))] | \le \eps.$$
\end{lemma}
Therefore, 
to ensure the generalization bound of a given random algorithm $\mcl A$, we simply compute the uniform stability bound $\eps$. 
To establish this result, we additionally require Assumption~\ref{assumption:bdd-grad}, 
i.e. boundedness of the local gradients.

The following theorem 
summarizes our result of uniform stability for DFedAvgM given the assumptions stated previously; the proof can be found in the appendix. 

\begin{theorem}[Uniform stability]\label{thm:main}
    Under Assumptions~\ref{assumption:Lsmooth}-\ref{assumption:bdd-grad}, we have that for any $T$ if the step size $\eta_t \le \frac{c}{t}$ and $c$ is small enough, then DFedAvgM satisfies uniform stability with
    \begin{equation}\label{eq:stab-bound}
    {\small \begin{aligned}
                \epsilon \le T^{\frac{cLK}{1 + cLK}} \lp \frac{(\sup f )K(cLK)^{\frac{1}{1 + cLK}}}{n} +  \frac{\frac{2\sigma B}{NL}}{(cLK)^{\frac{cLK}{1 + cLK}}}\rp
                + \frac{B(\sigma + B) \lp cK + 2C_\lambda \rp}{cLK},
    \end{aligned}}
    \end{equation}
    where $\sup f < \infty$ is the uniform bound on the size of the non-negative global loss function $f$, $n$ is the local data set size, and 
    \[ 
        C_\lambda := 2\lambda^2 +  4\lambda^2 \ln \frac{1}{\lambda} + 2\lambda + \frac{2}{\ln \frac{1}{\lambda}}
    \]
    is a constant depending on the graph topology.
\end{theorem}

Per Lemma~\ref{lemma:gen-unif-stab}, we have that the generalization error for DFedAvgM is bounded by 
the same constant that bounds the uniform stability, $\epsilon$. 
Again, we note here the {\it explicit dependence} of the generalization error on the corresponding value of $\lambda$ for the mixing matrix $\mM$ 
of the graph topology. $C_\lambda$ is an increasing function of $\lambda \in (0,1)$ which implies that the bound in 
\eqref{eq:stab-bound} improves with smaller $\lambda$.

\section{Network Topology Design}\label{sec:network}~
The results of Section~\ref{sec:theory} show that the network topology has a profound impact on both optimization and generalization of DFedAvgM. 
According to Theorems~\ref{thm:gen-nonconvexity} and \ref{thm:main}, the closer $\lambda$ is to $1$ the slower DFedAvgM converges (Theorem~\ref{thm:gen-nonconvexity}) and the worse it 
generalizes (Theorem~\ref{thm:main}). To improve DFedAvgM, we propose a theoretically efficient and practical \emph{sparse network topology} whose $\lambda$ is far away from $1$.



For the sake of notation, we recall graph definitions and properties to introduce network construction. Given an undirect, connected graph $\mcl G = (\mcl V, \mcl E)$ we define the {\it graph Laplacian} $\mL =\mD - \mA$, 
where $\mA = [a_{i,j}]$ (with $a_{i,j} = 1$ if $(i,j) \in \mcl E$) is the adjacency matrix and $\mD_{i,i} = \sum_{j=1}^N a_{i,j}$ is the diagonal degree matrix of $\mcl G$. Since $\mcl G$ is undirected, we have that both $\mA$ and $\mL$ are symmetric. 
Note that $\mL$ is positive semidefinite, with a trivial eigenvalue of $0$ occurring with multiplicity reflecting the number of connected components in $\mcl G$. As we assume that $\mcl G$ is connected, this means that only the first eigenvalue $\lambda_1(\mL) = 0$, and we can order the rest of the eigenvalues as $0 = \lambda_1(\mL) < \lambda_2(\mL) \le \lambda_3(\mL) \le \ldots \le \lambda_N(\mL)$. Define the {\it reduced condition number} of $\mL$ as
\begin{equation}
    \kappa(\mL) := \frac{\lambda_{N}(\mL)}{\lambda_2(\mL)},
\end{equation}
which is a measure of graph connectivity because a {\it smaller} $\kappa(\mL)$ corresponds to a graph with {\it higher} connectivity. This is an important constant that allows us to quantify how useful a given graph topology is for the purposes of improving convergence and generalization of DFedAvgM.


We apply the mixing matrix used 
in \cite{chow2016expander} 
\[
    {\mM} = \mI - \frac{2}{(1 + \theta)\lambda_N(\mL)} \mL, \quad \theta \in [0, 1),
\]
which allows us to quantify the associated value of $\lambda$. 
The eigenvalues of this mixing matrix ${\mM}$ have a straightforward relationship with eigenvalues of $\mL$:
\[
    \lambda_i({\mM}) = 1 - \frac{2}{(1 + \theta) \lambda_m(\mL)} \lambda_i(\mL).
\]
Then it is clear that
{ \begin{align*}
    \lambda &= \max \{ |\lambda_2({\mM})|, |\lambda_m({\mM} )|\} \\
    &= \max \left\{ \left|1 - \frac{2}{(1 + \theta) \lambda_m(\mL)} \lambda_2(\mL)  \right|, \left| 1 - \frac{2}{(1 + \theta) \lambda_m(\mL)} \lambda_m(\mL) \right| \right\} \\
    &= \max \left\{ \frac{\left|1 + \theta - \frac{2}{\kappa(\mL)}\right|}{1 + \theta}   ,  \frac{1 - \theta}{1 + \theta} \right\}.
\end{align*}}


\begin{figure}[!ht]
\centering
\begin{tabular}{c}
\includegraphics[width=0.7\linewidth]{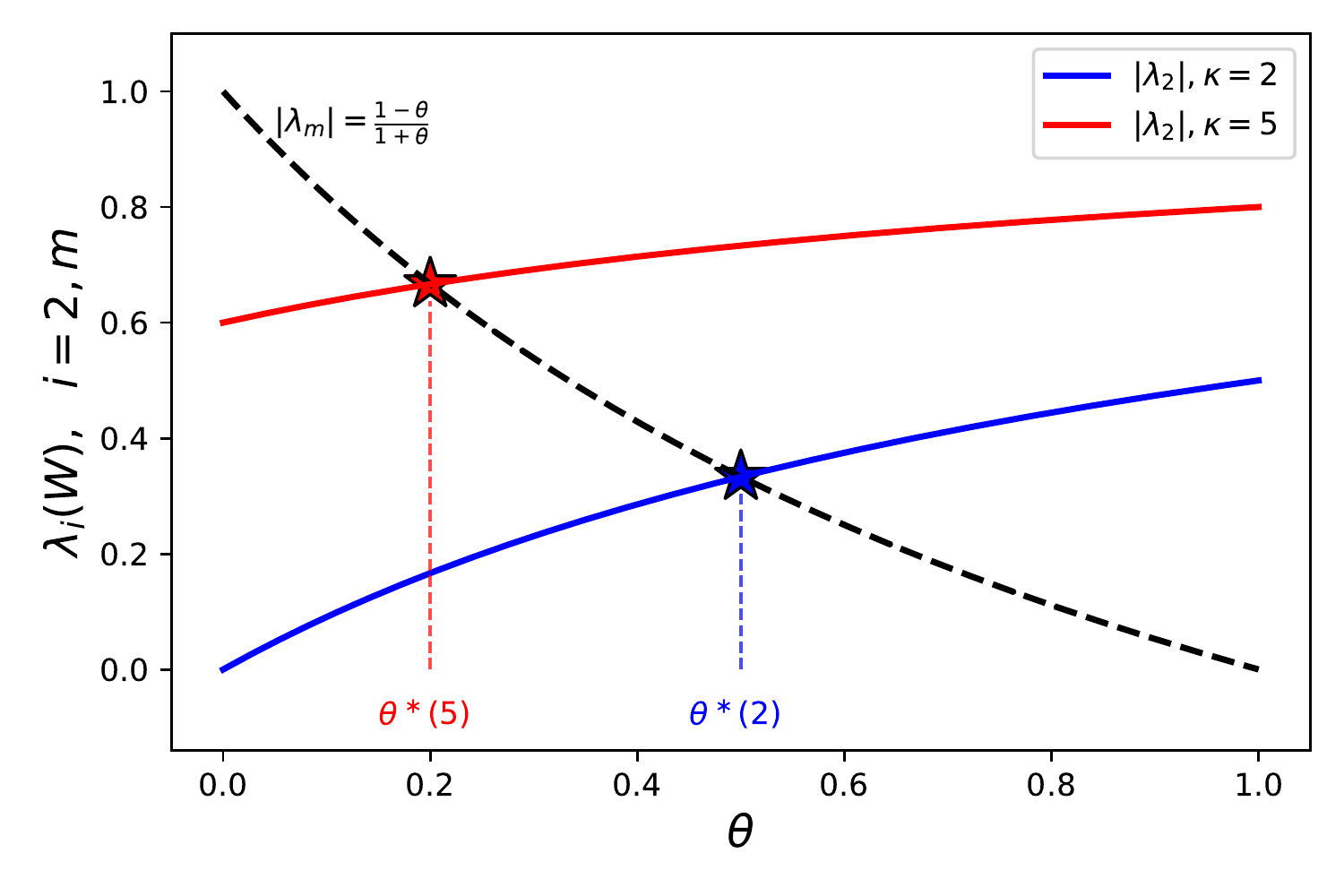}\\
\end{tabular}\vspace{-0.5cm}
\caption{Plot of $\lambda$
for the expander graph as a function of $\theta$. The stars indicate the optimal choice of $\theta = \theta^\ast(\kappa)$, given the value of $\kappa(\mL)$. {\it Lower} $\kappa$ leads to a {\it higher} value of $\theta^\ast(\kappa)$, which in turn leads to a {\it lower} value of $\lambda$, which is desired.}\label{fig:lambda-plot}
\end{figure}

For a fixed $\kappa(\mL)$, we can view $\lambda$ as a function of $\theta$ that can be optimized to lead to the lowest value of $\lambda$. In Fig.~\ref{fig:lambda-plot}, we have plotted the function $|\lambda_m(\mM)| = (1-\theta)/(1+ \theta)$ along with  $|\lambda_2(\mM)| = |1 + \theta - 2/\kappa(\mL)|/(1+\theta)$ for two values of $\kappa(\mL)$. For each fixed $\kappa(\mL)$, the corresponding lowest value of $\lambda(\theta)$ occurs when $|\lambda_2(\mM)| = |\lambda_m(\mM)|$, shown in Fig.~\ref{fig:lambda-plot} as stars; that is, when $\theta = \theta^\ast(\kappa(\mL))$. It is clear that
\begin{align*}
    |\lambda_2(\mM)| &= |\lambda_m(\mM)| \\ 
    \iff \frac{1 + \theta^\ast(\kappa(\mL)) - \frac{2}{\kappa(\mL)}}{1  + \theta^\ast(\kappa(\mL))} &= \frac{1 - \theta^\ast(\kappa(\mL))}{1 + \theta^\ast(\kappa(\mL))} \\
    \iff \theta^\ast(\kappa(\mL)) - \frac{2}{\kappa(\mL)} &= -\theta^\ast(\kappa(\mL)) \\
    \iff \theta^\ast(\kappa(\mL)) &= \frac{1}{\kappa(\mL)},
\end{align*}
as long as $ 2/\kappa(\mL) - \theta^\ast(\kappa(\mL)) \le 1$, which is reasonable since most $\kappa(\mL) \ge 2$.

It is straightforward then that choosing a graph structure with a smaller value of $\kappa(\mL)$ gives $\mcl G$ better connectivity properties. Somewhat in competition with this connectivity is the communication cost of a given graph topology; that is, better connectivity of a graph structure generally corresponds to more edges in the graph which increases the 
communication cost. Each node sends its updated model 
to each of its neighbors, and so an increased number of edges results in more communication that must happen between nodes. 

\textbf{\it We propose using $d$-regular expander graphs to balance this connectivity communication tradeoff}. A $d$-regular graph has a fixed number degree $d$ for each node; i.e. $d(i)=d$ for all $i$. Expander graphs are sparse graphs that have strong connectivity properties, of which $d$-regular expander graphs (and the special case of {\it Ramanujan graphs}) are in a sense ``optimal'' graph connectivity structures (captured in the constant $\kappa(L)$) with fixed communication cost. While Ramanujan graphs are not known for every value of total nodes $N$ and degree $d$, with high probability most $d$-regular graphs are approximately Ramanujan for large enough $N$~\cite{chow2016expander}. 

For $d$-regular graphs, 
there exists a convenient upper bound for $\kappa(\mL)$. This bound involves the first non-trivial eigenvalue, $\lambda_1(\mA)$, of the corresponding adjacency matrix 
\begin{equation*}
    \kappa(\mL) \le \frac{d + \lambda_1(\mA)}{d - \lambda_1(\mA)}.
\end{equation*}
If the $d$-regular graph in question is Ramanujan, then we can bound $\kappa(\mL)$ as
\begin{equation}\label{eq:ramanujan-kappa}
    \kappa^R(\mL) \le \frac{d + 2\sqrt{d-1}}{d - 2\sqrt{d-1}}.
\end{equation}

As the right-hand side of (\ref{eq:ramanujan-kappa}) is a {\it decreasing} function of $d$, this would suggest to choose larger $d$ in order to minimize $\kappa^R(\mL)$. However, increasing $d$ will incur greater communication costs. One can in practice choose the value of $d$ according to a prescribed bound on the total communication cost.


\subsection{Comparison to Ring and Erd\"{o}s-R\'enyi graphs}\label{sec:suboptimal-ring-er}~
We show that other graph topologies are in a sense suboptimal for the purposes of DFedAvgM, highlighting two common examples: Ring and Erd\"{o}s-R\'enyi graphs. We emphasize that using {\it $d$-regular Ramanujan graphs are in a sense ``optimal''} by possessing {\it strong connectivity} properties in the graph topology while requiring {\it low communication cost} for local node neighborhood communication (i.e. sparsity).

\paragraph{Ring graphs -- poor connectivity.}

Ring graph is an extremely sparse, but still connected, $2$-regular graph structure where the graph structure constitutes a ring (see Fig.~\ref{fig:graphs} (b)). While a very simple and sparse topology to impose on the nodes of the graph, ring graphs possess poor connectivity properties that we can directly compare with $d$-regular Ramanujan graphs.

It is well-known that the eigenvalues of the graph Laplacian 
$\mL_{ring}$ of the ring graph on $N$ nodes are given by
\[
    \{\lambda_k(\mL_{ring})\}_{k=0}^{\frac{N}{2}} = \left\{  2 - 2 \cos \left(\frac{2\pi k}{N}\right) \right\}_{k=0}^{\frac{N}{2}}, 
\]
each with geometric multiplicity $2$, except for the first eigenvalue $\mu_0(\mL_{ring}) = 0$ which has geometric multiplicity 1; if $N$ is even, the last eigenvalue $\lambda_{N/2}(\mL_{ring})$ has multiplicity 1 as well. Therefore, we can straightforwardly see that the reduced condition number for a Ring graph on $N$ nodes is
\begin{align*}
    \kappa^{ring}(\mL_{ring}) &= \frac{\lambda_{N/2}(\mL_{ring})}{\lambda_{1}(\mL_{ring})} = \frac{2 - 2 \cos\left(\frac{2\pi N}{2N}  \right)}{2 - 2 \cos\left( \frac{2\pi}{N} \right)}\\
    &= \frac{4}{2 - 2 \cos\left(\frac{2\pi}{N}\right)} \ge \frac{4}{2 - 2 \left(1 - \frac{1}{2}\left(\frac{2\pi}{N}\right)^2  \right)}
    = \frac{4N^2}{4 \pi^2} = \frac{N^2}{\pi^2}.
\end{align*}

Therefore, we see that with this {\it lower bound} for the Ring graph has a $\kappa(\mL_{ring})$ grows {\it quadratically} with the size $N$ of the graph! The corresponding value for $\lambda$ approaches 1 for increasing values of $N$, which implies slower convergence rates (Theorem~\ref{thm:gen-nonconvexity}) and worse generalization bounds (Theorem~\ref{thm:main}). It is clear then that
\[
    \kappa^R(\mL) \le \frac{d + 2\sqrt{d-1}}{d - 2\sqrt{d-1}} \ll \frac{N^2}{\pi^2} \le \kappa^{ring}(\mL_{ring}),
\]
which shows superior convergence properties of Ramanujan expander graphs compared to the sparse Ring graph structure.

\paragraph{Erd\"{o}s-R\'enyi graph -- high communication cost.}
Another type of graph topology one could impose for DFedAvgM is an Erd\"{o}s-R\'enyi (ER) random graph structure, wherein each edge $(i,j) \in \mcl V \times \mcl V$ is sampled independently and indentically distributed with probability $p \in (0,1)$. It is well-known that as long as $p = \mcl O(\ln N/N)$, then the resulting graph $\mcl G$ is connected with high probability~\cite{erdos-renyi1960}. 

While the connectivity properties of ER graphs are nearly guaranteed to be better than $d$-regular Ramanujan graphs (with $d$ is relatively small), {\it the communication cost of ER graphs is prohibitively large} for large network size $N$. To see this, the expected degree $d_i$ of a node $i \in \mcl V$ in an ER graph with large enough edge probability $p$ is simply $\bar{d} = Np = \mcl O(\ln N)$, which grows with the size of the graph $N$. This incurs a much larger communication cost than the constant cost of $d$-regular expander graphs as it is assumed that $d \ll N$, with $d < \ln N$ as $N$ is large.

In Section~\ref{sec:experiments}, 
we empirically verify the superior {\it connectivity-communication cost} balance exemplified by the $d$-regular expander graph structure compared to Ring and ER graphs for DFedAvgM. These $d$-regular expander graphs have better connectivity properties than Ring graphs while at the same time being sparser (i.e. lower communication costs) than ER graphs.


\section{Practical network design}\label{sec:practical-network}~
In this section, we discuss how to convert a given graph to a practical overlay network topology 
for DFL. 
We illustrate our proposed $d$-regular network topology in Fig.~\ref{fig:DFL-topology}: for $d$-regular graph suppose $d$ is even and let $L=d/2$,
we assign for each node 
a set of \emph{virtual coordinates} represented by a $L$-dimensional vector $\langle x_{1}, x_{2}, ..., x_{L} \rangle$, where each element $x_{i}$ is a randomly generated real number $0\leq x_{i}<1$, as shown in Fig.~\ref{fig:DFL-topology} (a).
%
There are $L$ virtual ring spaces such as the two shown in Fig.~\ref{fig:DFL-topology} (b). In the $i^{th}$ space, a node is \emph{virtually} placed on a ring based on the value of its $i^{th}$ coordinate $x_i$.
Coordinates in each space are circular, and 0 and 1 are superposed.
For each space, a node will connect to the two adjacent nodes, for example, $B$ connects to $A$ and $C$ in Space 1 and $G$ and $F$ in Space 2. Hence each node has at most $d=2L$ neighbors.  A neighbor of a node may happen to be adjacent to it in multiple spaces, such as $A$ and $D$. In such 
case, $A$ can connect to another node in the same situation, such as $E$. In the end, the equivalent network topology is shown in Fig.~\ref{fig:DFL-topology} (c).

\begin{figure*}[!ht]
\centering
\includegraphics[width=12cm]{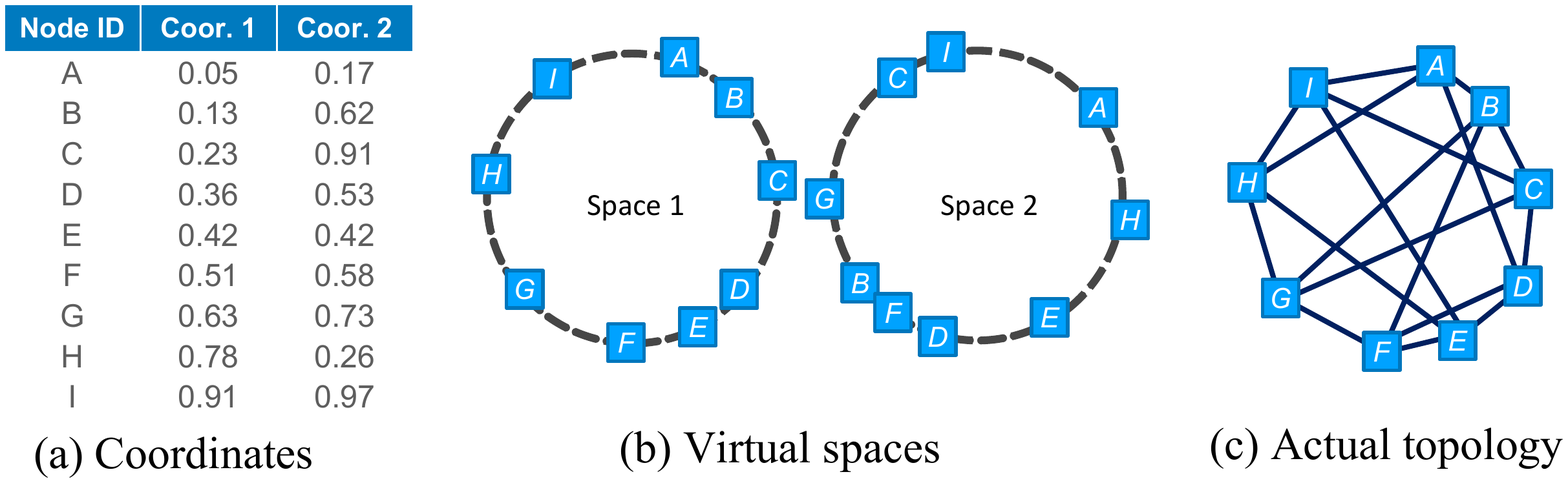}
\vspace{-4mm}
\caption{DFL network topology in working systems. Each node generates a set of coordinates and the network is generated in a distributed manner by allowing each node to execute the proposed protocols locally.}
\label{fig:DFL-topology}
\vspace{-2ex}
\end{figure*}

The proposed network is a close proximal construction for a random $d$-regular network \cite{S2-TPDS}.
Note that in practice there does not exist a perfect construction of a random $d$-regular graph \cite{Jellyfish}, and there is no way for a network node to verify whether the entire network is Ramanujan only based on its local information.  

The construction of a correct topology can be achieved by allowing each node to maintain the two closest nodes on each virtual ring. When a new node joins the network, it can always succeed to find the two closest nodes on each virtual ring by recursive queries \cite{MDT}.

\subsection{Network recovery from node failures}~
To maintain a correct DFL topology for a dynamic set of nodes, protocols should be designed to recover errors from node failures and leaves. Here an error is defined as a node that has a wrong neighbor set compared to a correct DFL network topology. 
If a node $x$ fails from the network, in each virtual space $i$, its adjacent nodes $y_i$ and $z_i$ should remove $x$ from their neighbors and add each other as a new neighbor. To recover from such single-node failure, the proposed recovery protocol allows each node to store the IP addresses of the two-hop neighbors. Hence if a node is detected to fail, its two adjacent nodes can directly connect as new neighbors.

\section{Experimental Results}\label{sec:experiments}


\subsection{Convergence and 
generalization}~
We evaluate the communication round versus training loss, test loss, test accuracy, and the communication cost for Ring, Erd\"{o}s-R\'enyi, fully-connected, and  the proposed expander graphs. We pick $d=3$ regular expander graphs (called Ramanujan). The communication cost could be estimated by the model size. In all experimental settings, the topology is generated by a central server before the training starts and stored in each user, but the central host are not involved in the actual training process. The expander graph is generated by adding an extra edge on top of the Ring graph. The Erd\"{o}s-R\'enyi graph is generated by selecting random edges from all possible edges with the probability $p = \frac{\ln{N}}{N}$, where $N$ is the total number of 
expander graphs result in 
faster convergence and better generalization of DFedAvgM in training different models on different datasets. To conduct more comprehensive and solid experiments and testing, both the real network settings and the simulation are used in our evaluation. To exclude other factors no tuned optimization and data compression algorithms are used in the experiments. In the evaluation, fully-connected graphs are shown as a baseline but it is hardly practical in real world applications considering the communication cost and availability. Ring topology is easy to implement and widely used in previous works, so it is also shown as a baseline. {Because of the randomness of the  Erd\"{o}s-R\'enyi graph, the experimental results are inconsistent when there are relatively few nodes; and so we do not include the Erd\"{o}s-R\'enyi graph in all MNIST experiments below. }

\paragraph{MNIST IID.}
We randomly split the MNIST dataset without any biases into $10$ different subsets. 
Each user owns a local multilayer perceptron (MLP) model with one hidden layer of size $200$. Each user only has 
access to only one local subset as its training set. We train the local model with the batch size of $20$ and use the cross entropy as the loss function. We use SGD with the learning rate 
$0.01$ and the momentum 
$0.9$. 
After $3$ epochs of local training, all the local nodes communicate with the topological neighbors and average all the parameters of the MLP model. After each communication round, the test accuracy, test loss, and training loss of each user are recorded and averaged in Fig.~\ref{fig:mnist-iid}. Based on our experiments, in this IID settings the fully-connected and expander graph converge at round $16$ which have advantage over $26$ rounds of the Ring graph. 
According to the test accuracy shown in Fig.~\ref{fig:mnist-iid}, the fully-connected graph has the best test accuracy of $98.2\%$ while the expander graph reaches a similar $98.0\%$ with only 
one third of its communication cost. The Ring graph reaches $97.7\%$ accuracy due to the ideal distribution of the data.

\begin{figure}[!ht]
\centering
\begin{tabular}{ccc}
\hspace{-0.2cm}\includegraphics[width=0.31\linewidth]{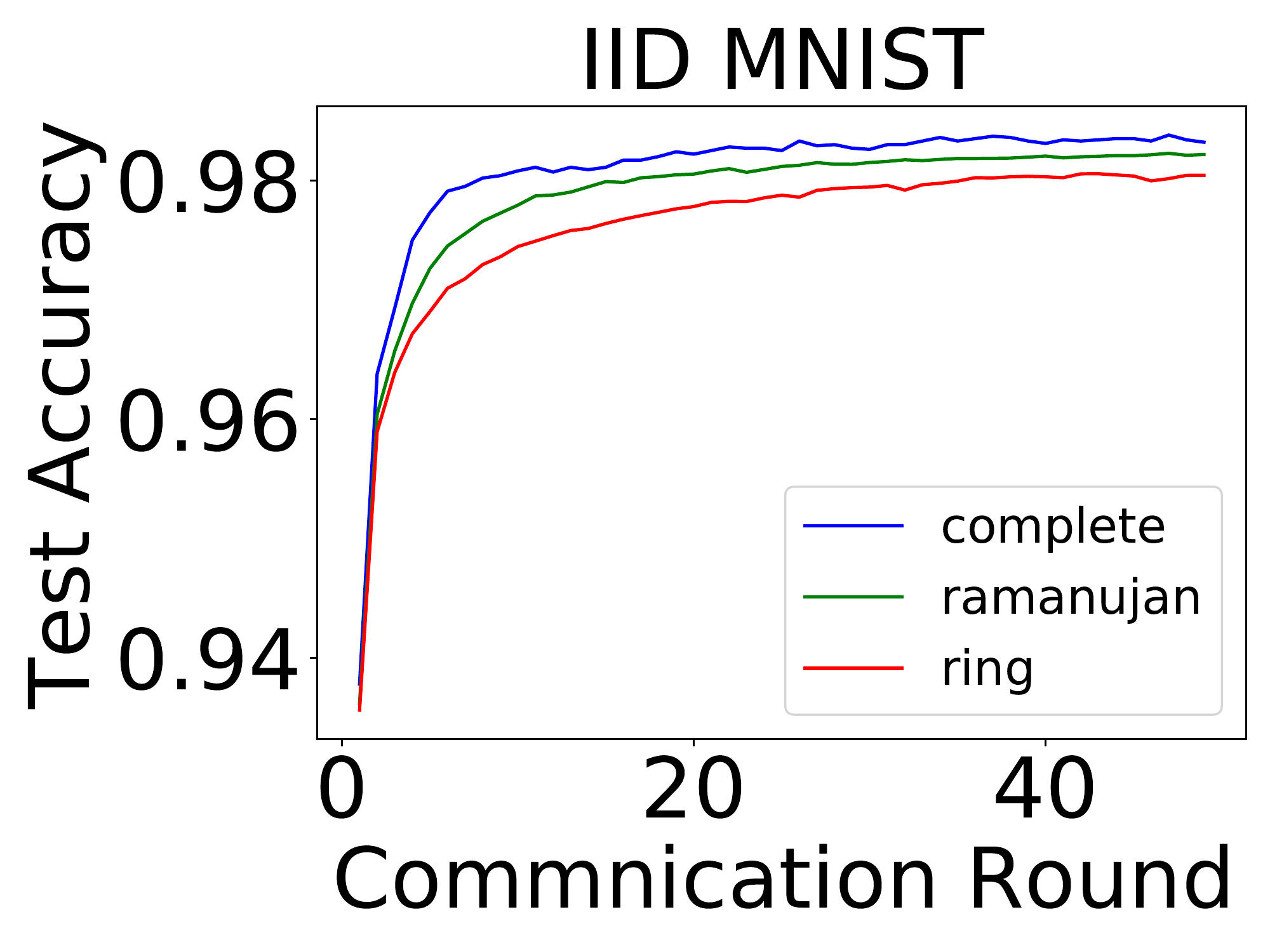}&
\hspace{-0.2cm}\includegraphics[width=0.31\linewidth]{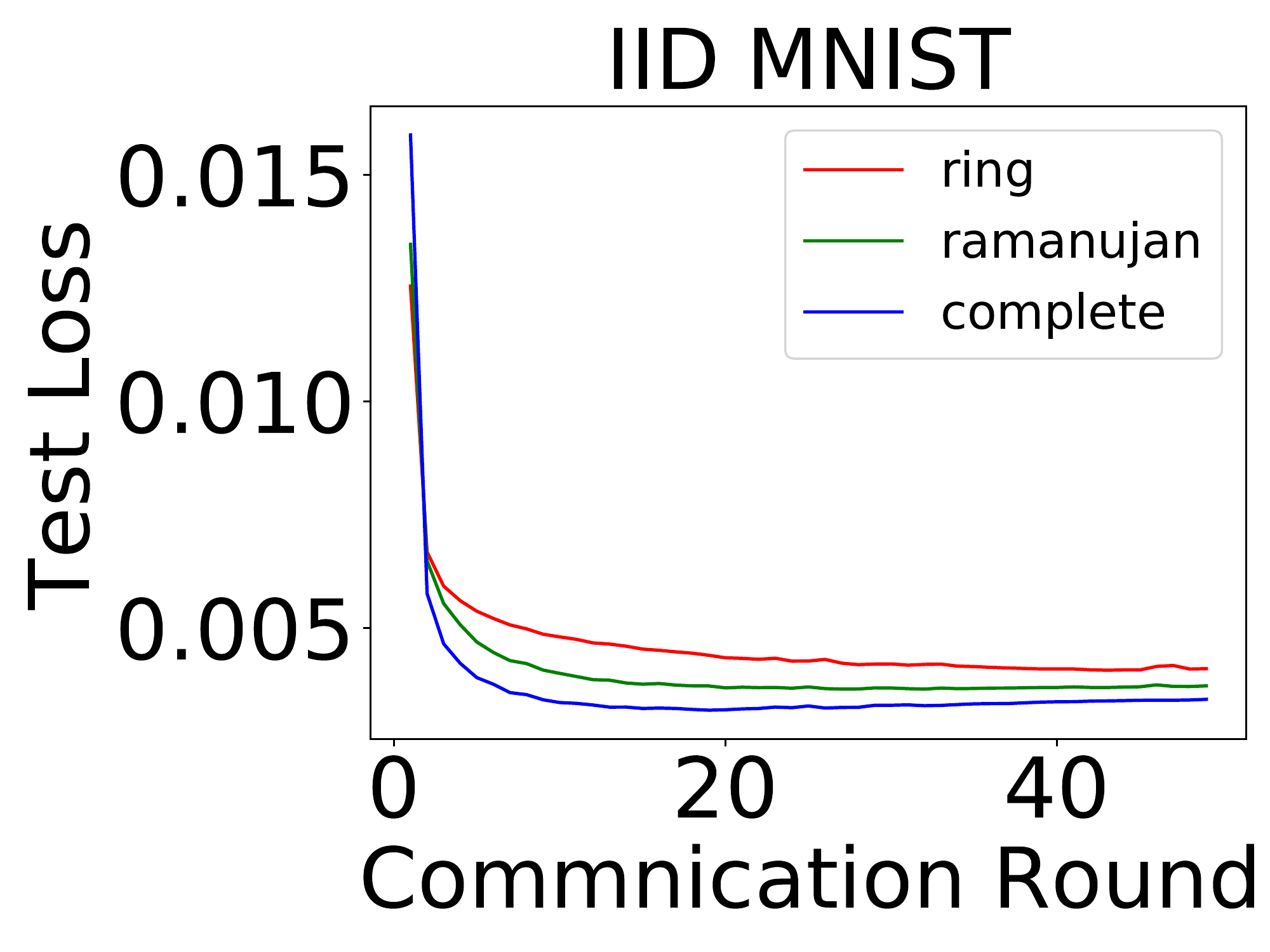}&
\hspace{-0.2cm}\includegraphics[width=0.31\linewidth]{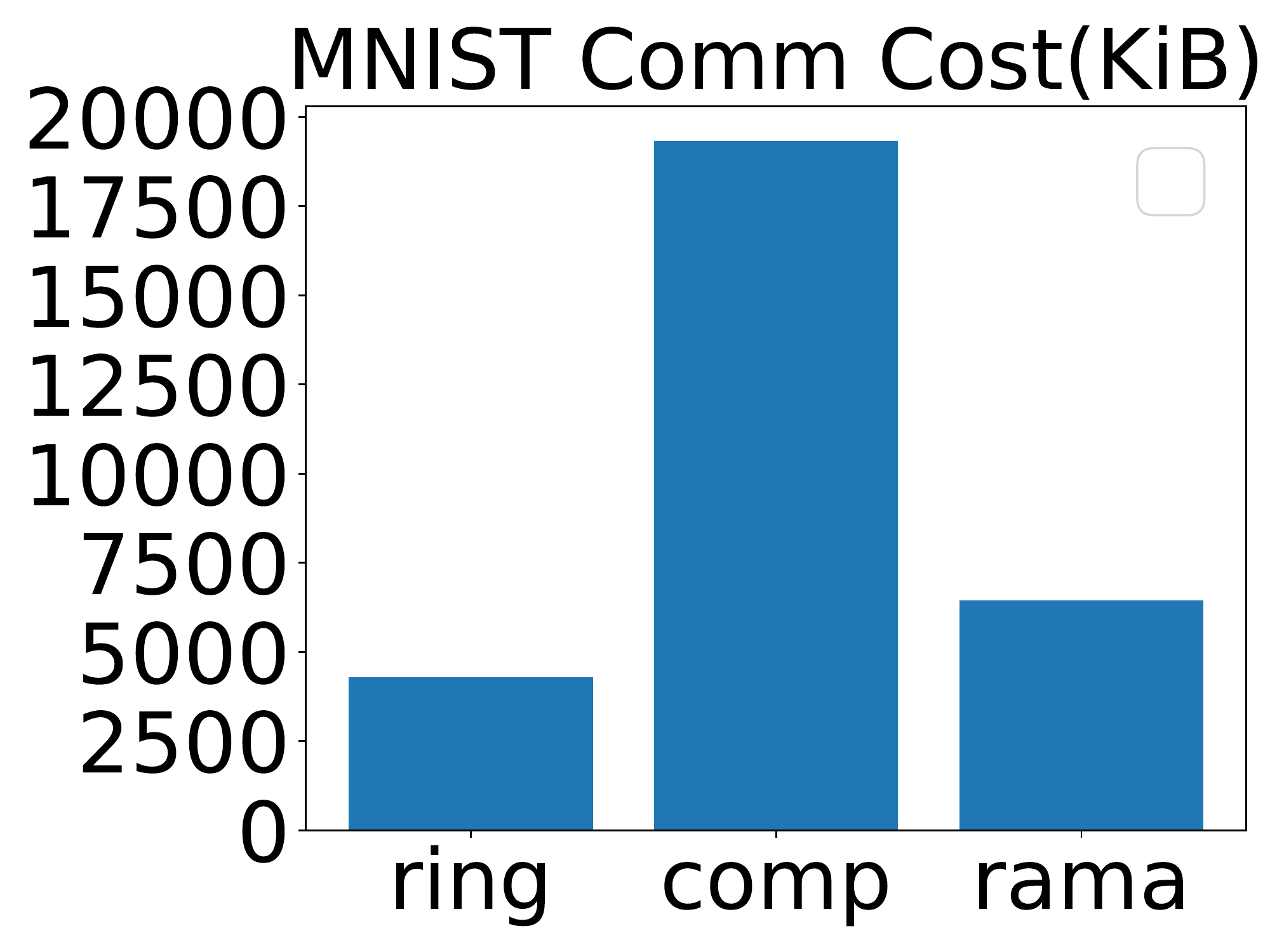}\\[-3pt]
{\footnotesize Test Accuracy} & {\footnotesize Test Loss} & {\footnotesize Comm. Cost}\\
\end{tabular}
\vspace{-0.4cm}
\caption{The test accuracy, test loss, 
and communication cost of the Ring/3-regular expander (Ramanujan)/Fully connected graphs on IID MNIST. All of the graphs reach over $92\%$ accuracy but the expander graph starts to converge at round 12 while the Ring graph starts to 
converge at round 20.
}
\label{fig:mnist-iid}
\end{figure}

\paragraph{MNIST Non-IID.}
All the settings are similar to the IID settings except each node owns a local dataset consisting of only one label (one 
digit in MNIST). The distribution is extremely unfavorable to the generalization. The test dataset is balanced sampled from the original dataset as the IID settings. As shown in Fig.~\ref{fig:mnist-noniid}
the expander graph reaches $88.8\%$ accuracy and much higher than the Ring graph ($73.68\%$). The fully connected graph reaches the best accuracy of $94\%$. Although the expander graph's accuracy is lower than the fully connected graphs' but with $33\%$ of its communication cost. After each communication round, the training and test loss, and test accuracy of each user are recorded and averaged (Fig.~\ref{fig:mnist-noniid}). The expander graph could achieve a faster convergence and better generalization than the Ring graph and the performance is close to the fully connected graph but with a more manageable communication cost.

\begin{figure}[!ht]
\centering
\begin{tabular}{ccc}
\hspace{-0.2cm}\includegraphics[width=0.31\linewidth]{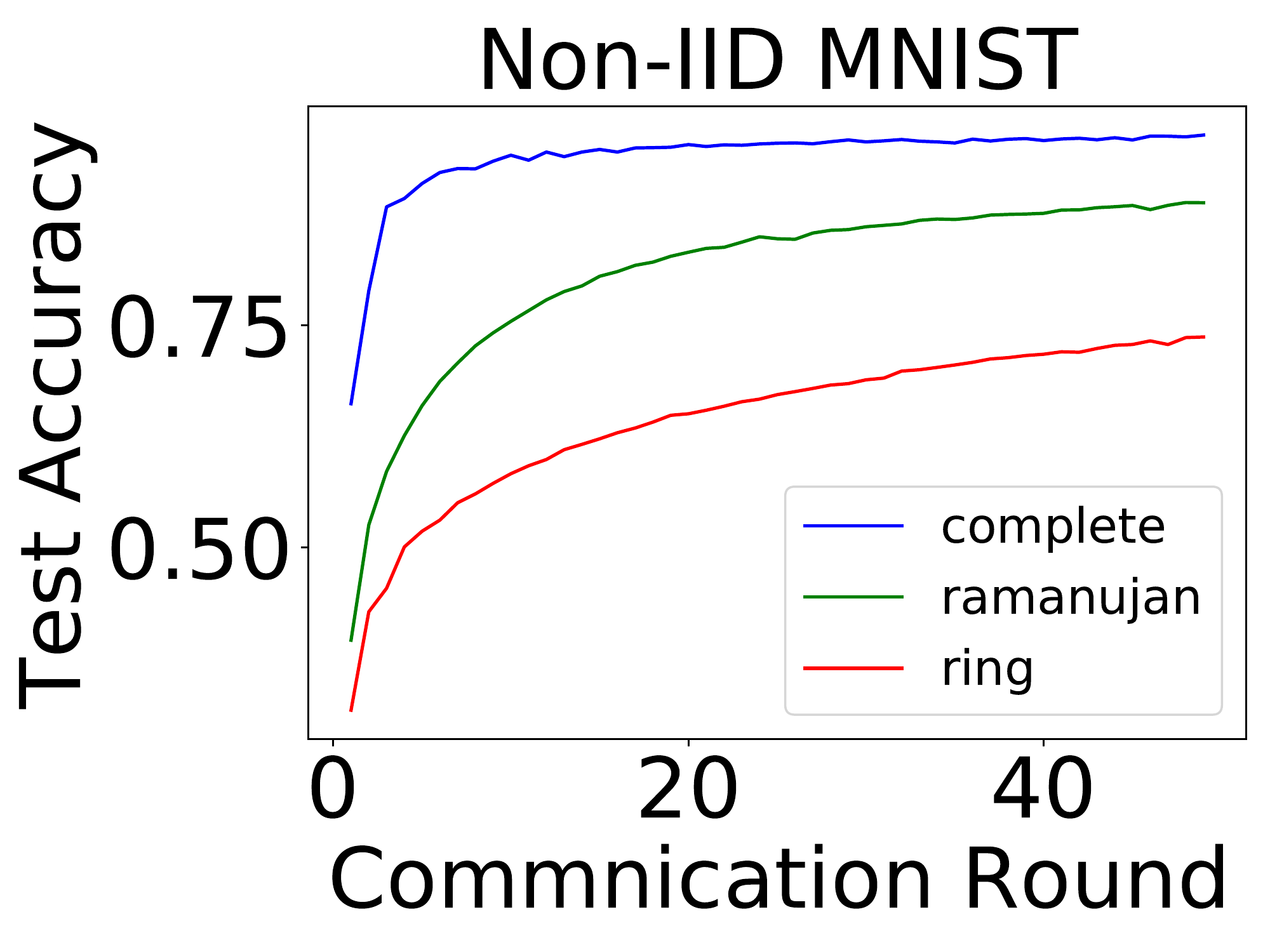}&
\hspace{-0.2cm}\includegraphics[width=0.31\linewidth]{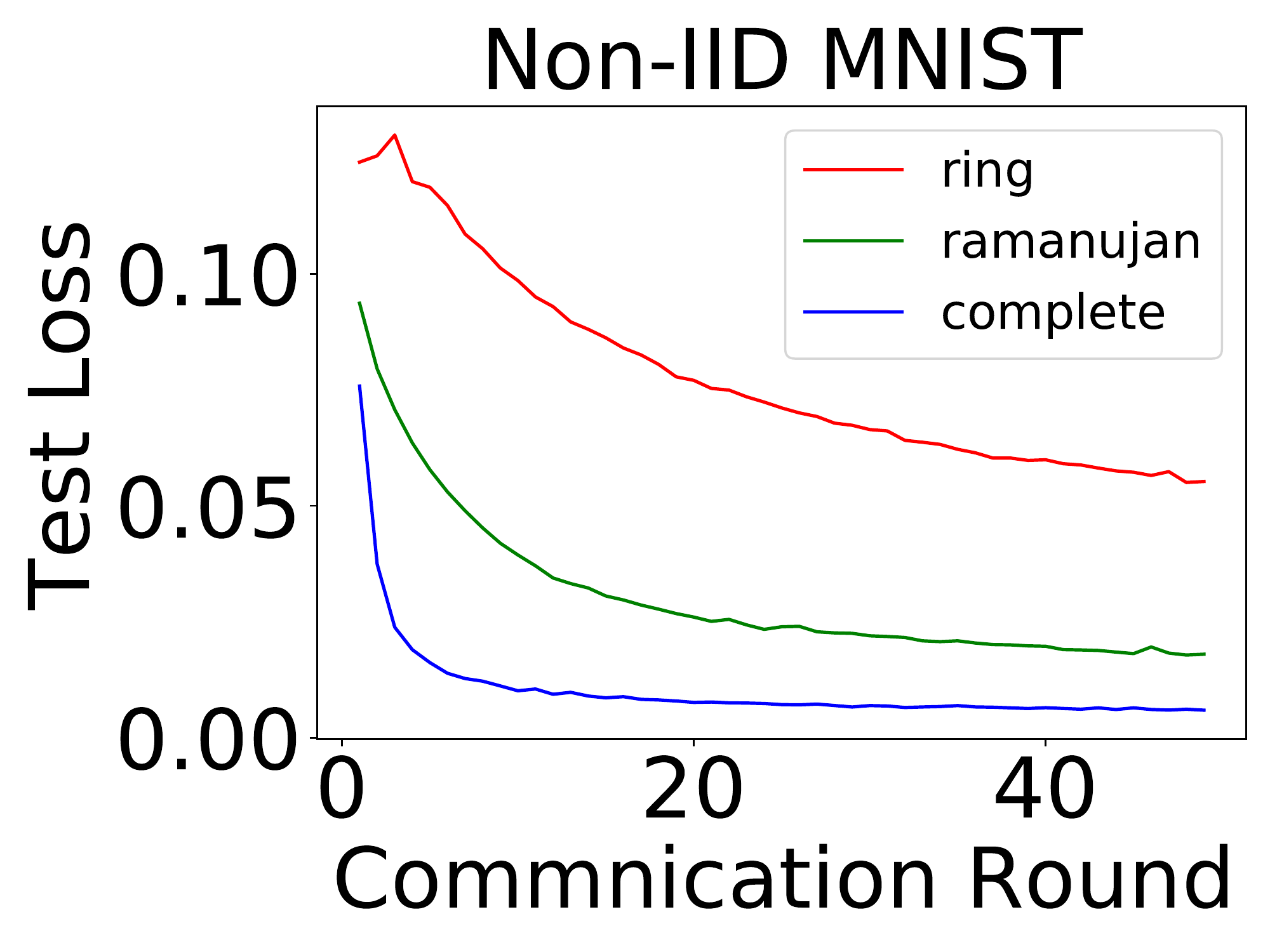}&
\hspace{-0.2cm}\includegraphics[width=0.31\linewidth]{img/mnist_comm_cost.pdf}\\[-3pt]
{\footnotesize Test Accuracy} & {\footnotesize Test Loss} & {\footnotesize Comm. Cost}\\
\end{tabular}
\vspace{-0.4cm}
\caption{The test accuracy, test loss, 
and communication cost of the Ring/3-regular Ramanujan/Fully connected Graph on non-IID MNIST. The expander graph reaches $88.79\%$ accuracy, which is higher than the Ring graph's $73.68\%$. The communication cost of the expander graph is only one third of the fully connected graphs'.
}
\label{fig:mnist-noniid}
\end{figure}\vspace{-0.2cm}

\paragraph{Language modeling.}
We further conduct the simulation to evaluate the 
effect of different topology to the language models. 
First, we split the Shakespeare dataset \cite{pmlr-v54-mcmahan17a} into 100 subsets (with some overlaps). Then we create $100$ LSTM models (each one with $256$ hidden units and $2$ layers). In this Non-IID sampling scenario, the underlying distribution of data for each node is consistent with the raw data. Since we assume that data distributions vary between users in the raw data, we take this setting as Non-IID. 
We use cross entropy as the loss function. Then we train each LSTM with the corresponding local Non-IID dataset with the 
learning rate 
$0.5$ and momentum 
$0.9$. After $3$ epochs of local training, all local nodes communicate with the neighbors through which is similar to the previous method employed by the MNIST experiment and average all the parameters of the LSTM model. After each communication round, the training loss, test loss, and test accuracy of each user are recorded and averaged in Fig.~\ref{fig:shakespeare}. The Erd\"{o}s-R\'enyi graph have an accuracy 
of $45.3\%$ which is close to the fully connected graph's $45.8\%$. The expander graph reach an accuracy 
of $40.4\%$ and the Ring graph only reaches $36.2\%$. In this unfavorable data distribution, the Ring graph generalize worse than the expander graph. 
{The Erd\"{o}s-R\'enyi graph has better test accuracy and test loss than the expander graph because it needs significantly more degrees to ensure the connectivity of the graph. Thus it has a higher communication cost.}
Also, DFedAvgM with the expander graph converges faster than the Ring graph. In this case, we could see that the communication cost for the complete graph is 16 times higher than the expander graph. With some moderate communication cost, expander graph could generalize better and has the similar convergence to the fully connected graph.

\begin{figure}[!ht]
\centering
\begin{tabular}{ccc}
\hskip-0.2cm\includegraphics[width=0.31\linewidth]{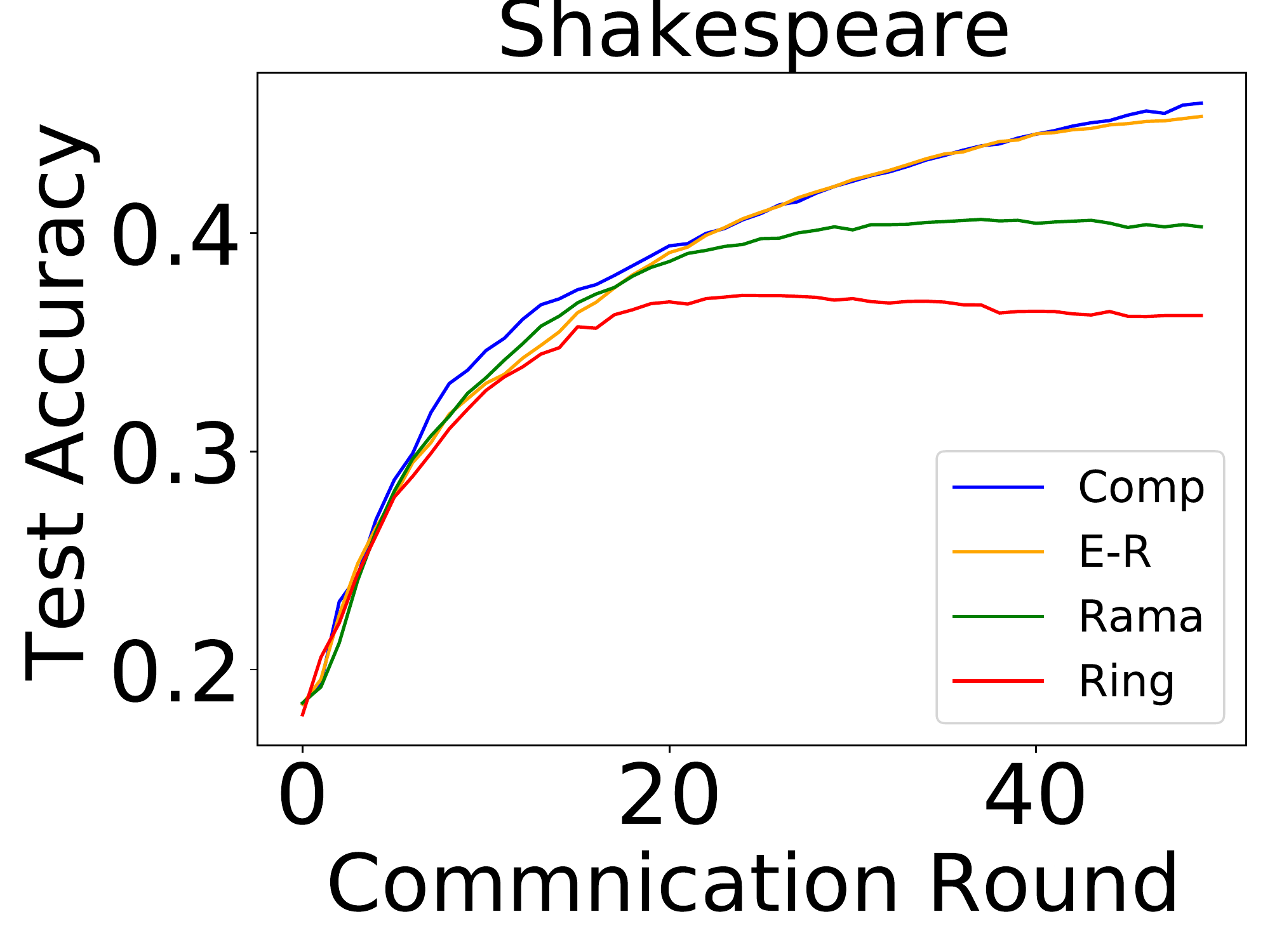}&
\hskip-0.2cm\includegraphics[width=0.31\linewidth]{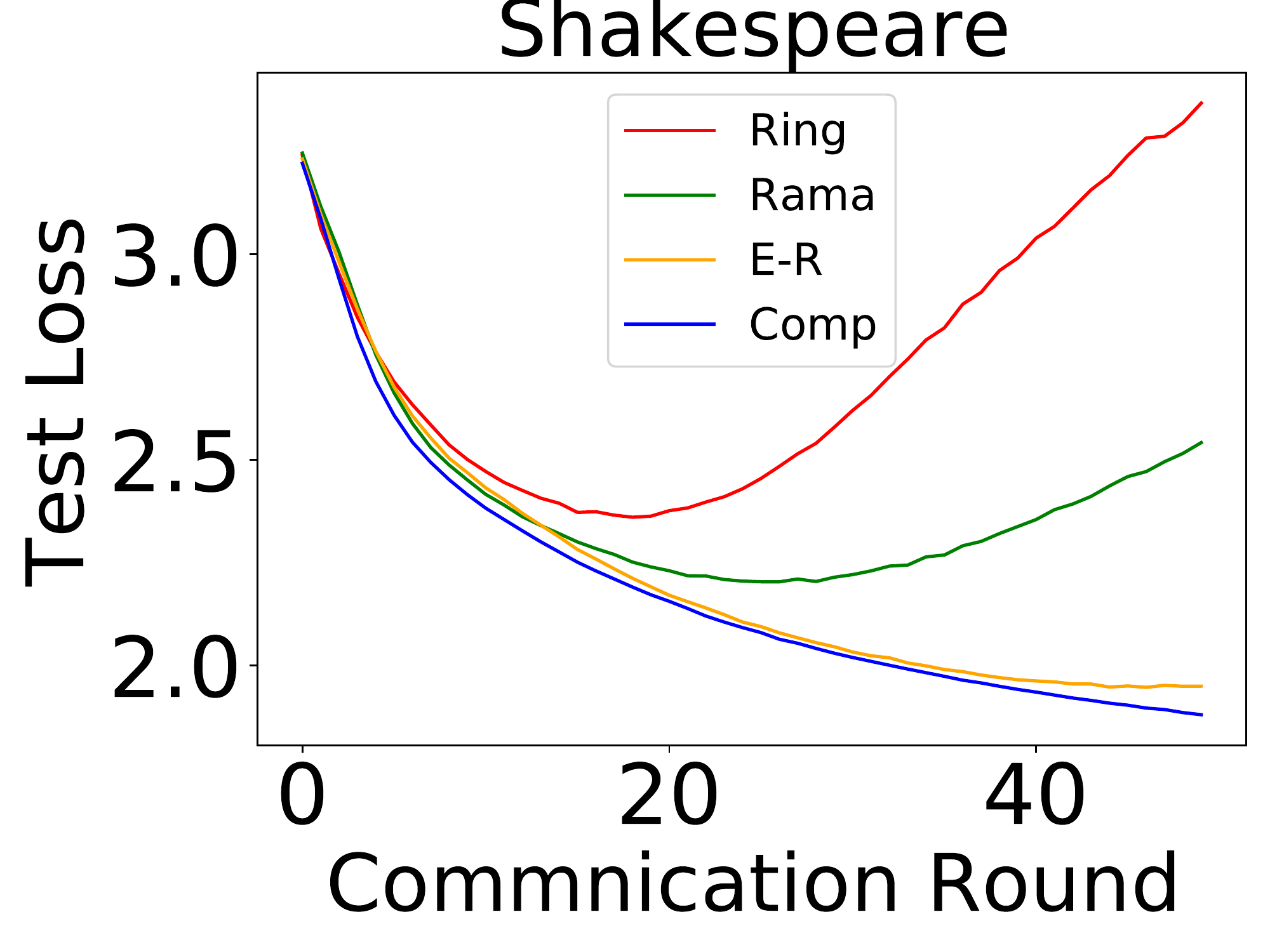}&
\hskip-0.2cm\includegraphics[width=0.31\linewidth]{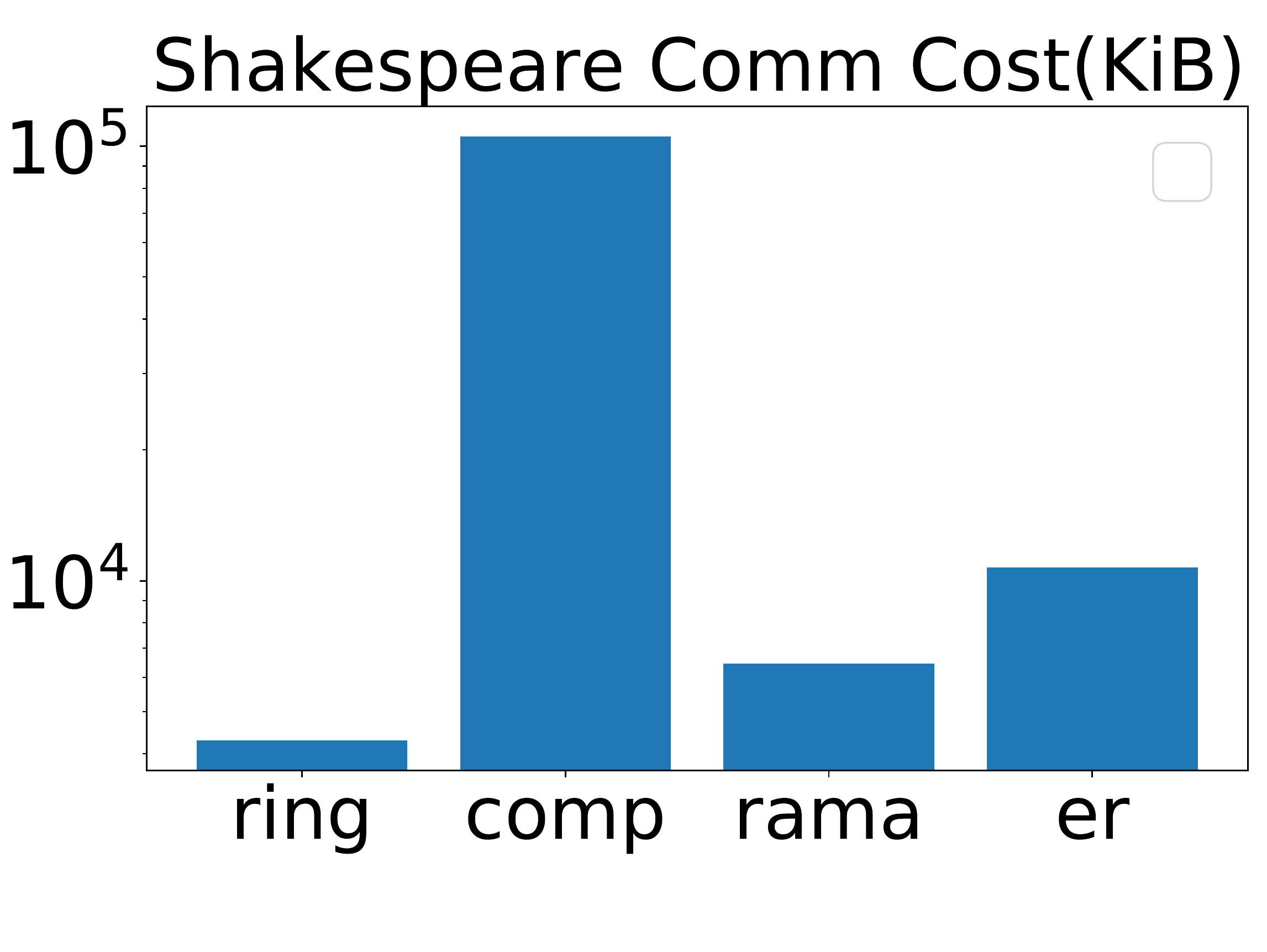}\\[-3pt]
{\footnotesize Test Accuracy} & {\footnotesize Test Loss} & {\footnotesize Comm. Cost}\\
\end{tabular}\vspace{-0.4cm}
\caption{The test accuracy, test loss, 
and 
communication cost of the Ring/3-regular Ramanujan/Erd\"os-R\'enyi/Complete graph on non-IID Shakespeare dataset.  Erd\"{o}s-R\'enyi graph have an accuracy rate of $45.3\%$. The expander graph reach a accuracy rate of $40.4\%$ and the Ring graph only reaches $36.2\%$.
}
\label{fig:shakespeare}
\end{figure}

\subsection{Robustness to client failures}~
To test the robustness of DFedAvgM with different network topology to 
client failures, 
we 
drop $10\%$ and $20\%$ of clients during the communication and compare the performance of Ring, expander, Erd\"{o}s-R\'enyi, and fully-connected graphs. In the language modeling, we mask the input of the dropped nodes to simulate the communication failure. 
All the dropped nodes are randomly selected and excluded from the final results.

\paragraph{MNIST Non-IID.}

\begin{figure}[!ht]
\centering
\begin{tabular}{cccc}
\hspace{-0.2cm}\includegraphics[width=0.23\linewidth]{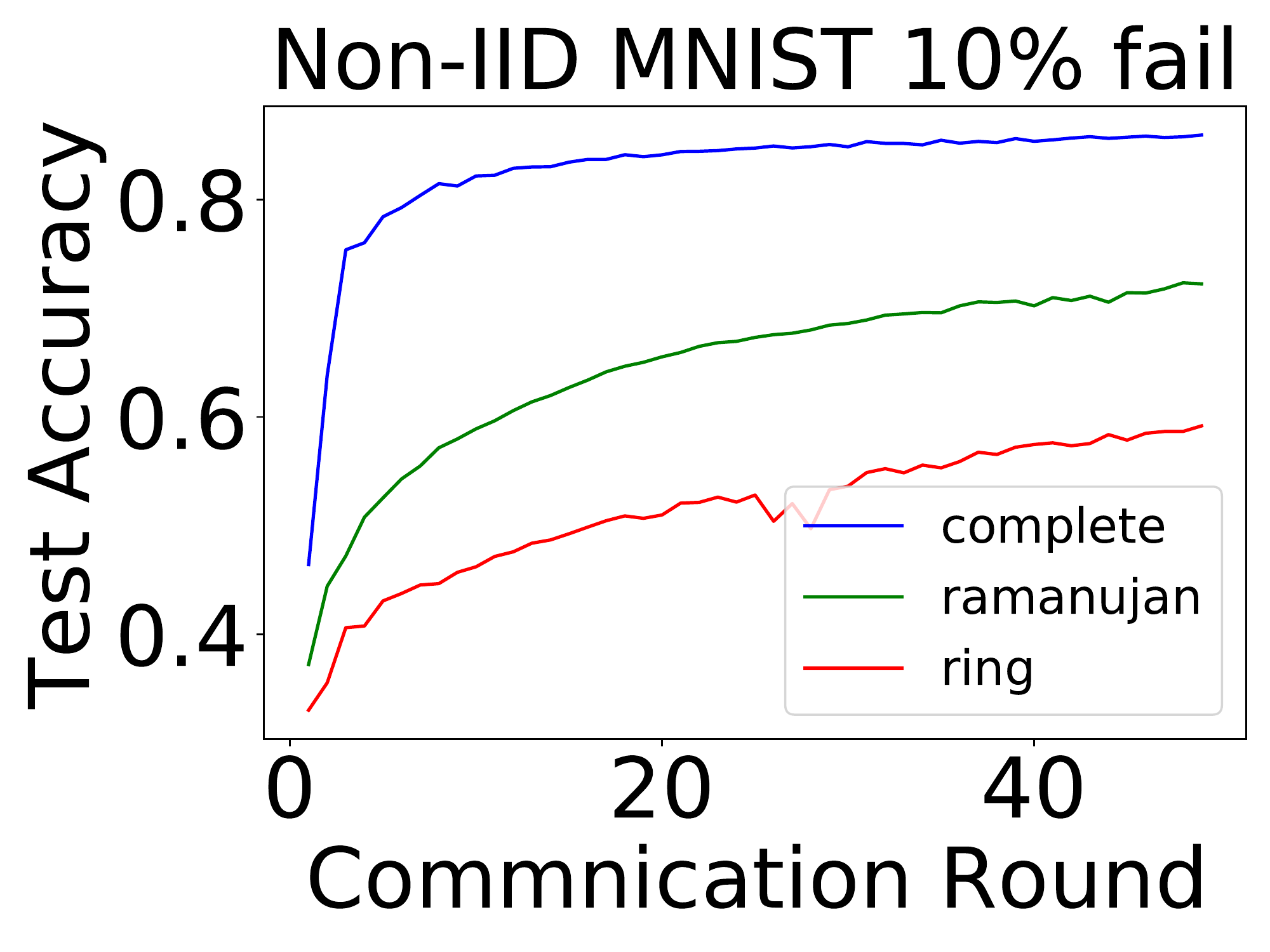}&
\hspace{-0.2cm}\includegraphics[width=0.23\linewidth]{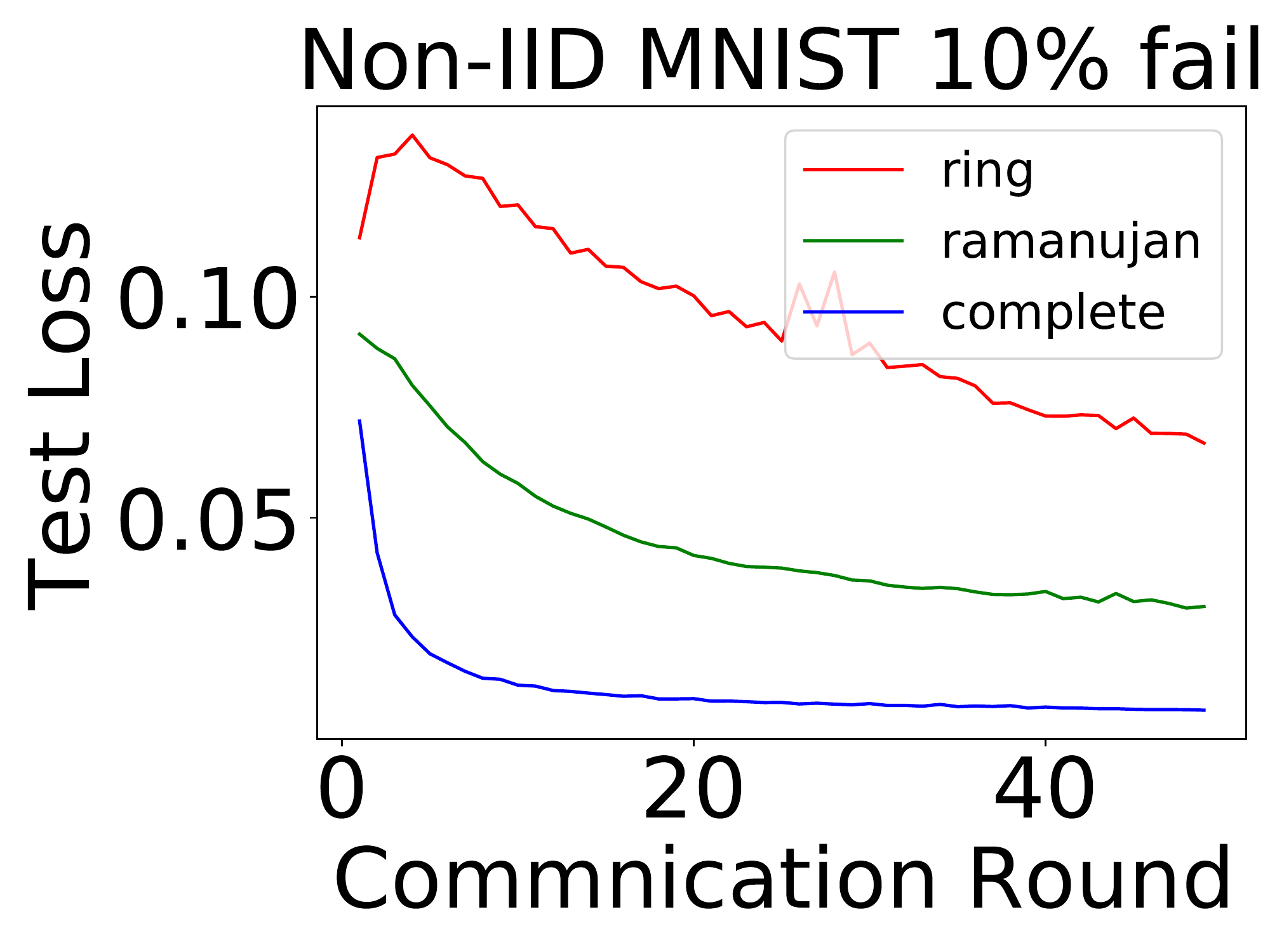}&
\hspace{-0.2cm}\includegraphics[width=0.23\linewidth]{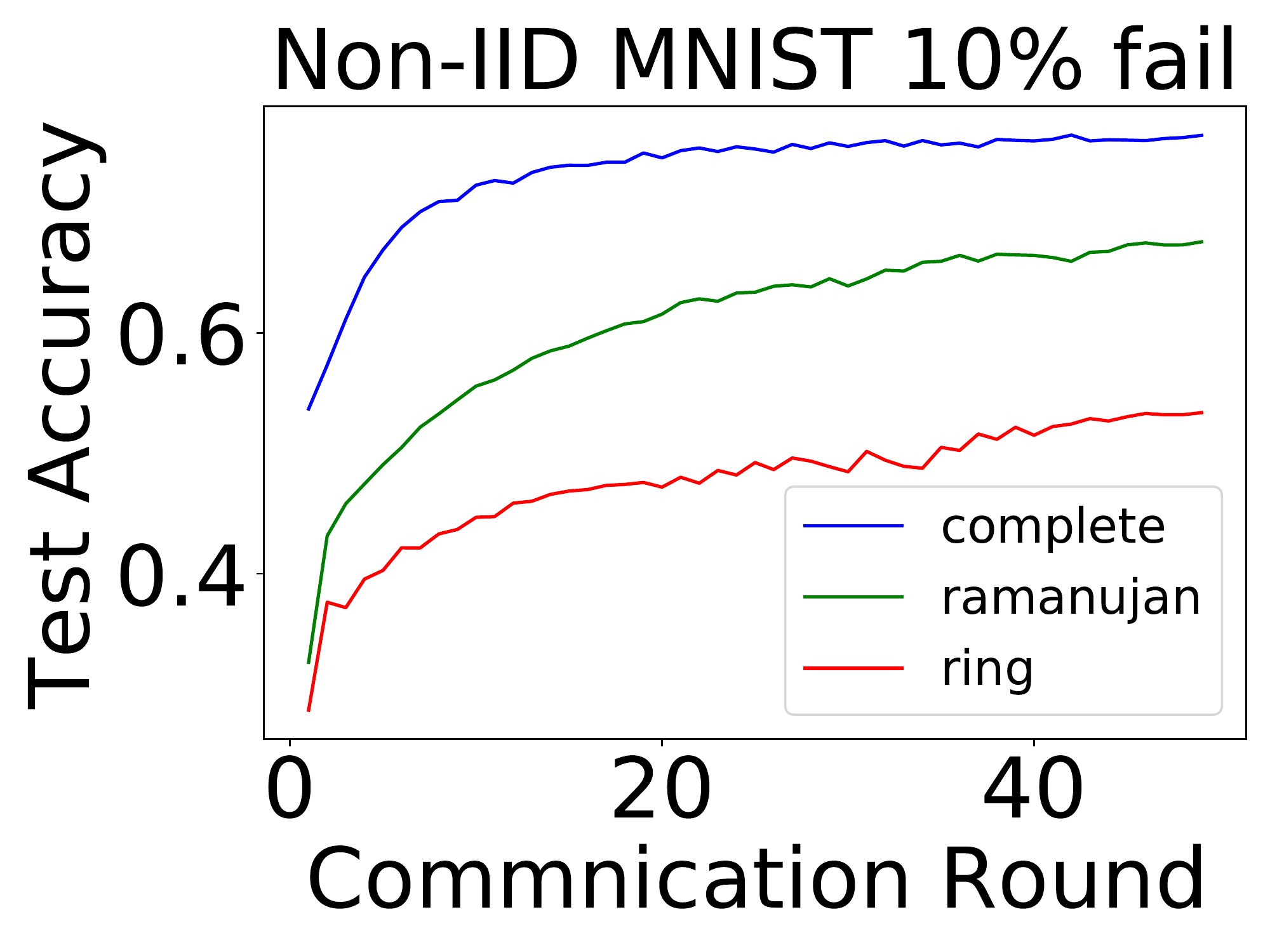}&
\hspace{-0.2cm}\includegraphics[width=0.23\linewidth]{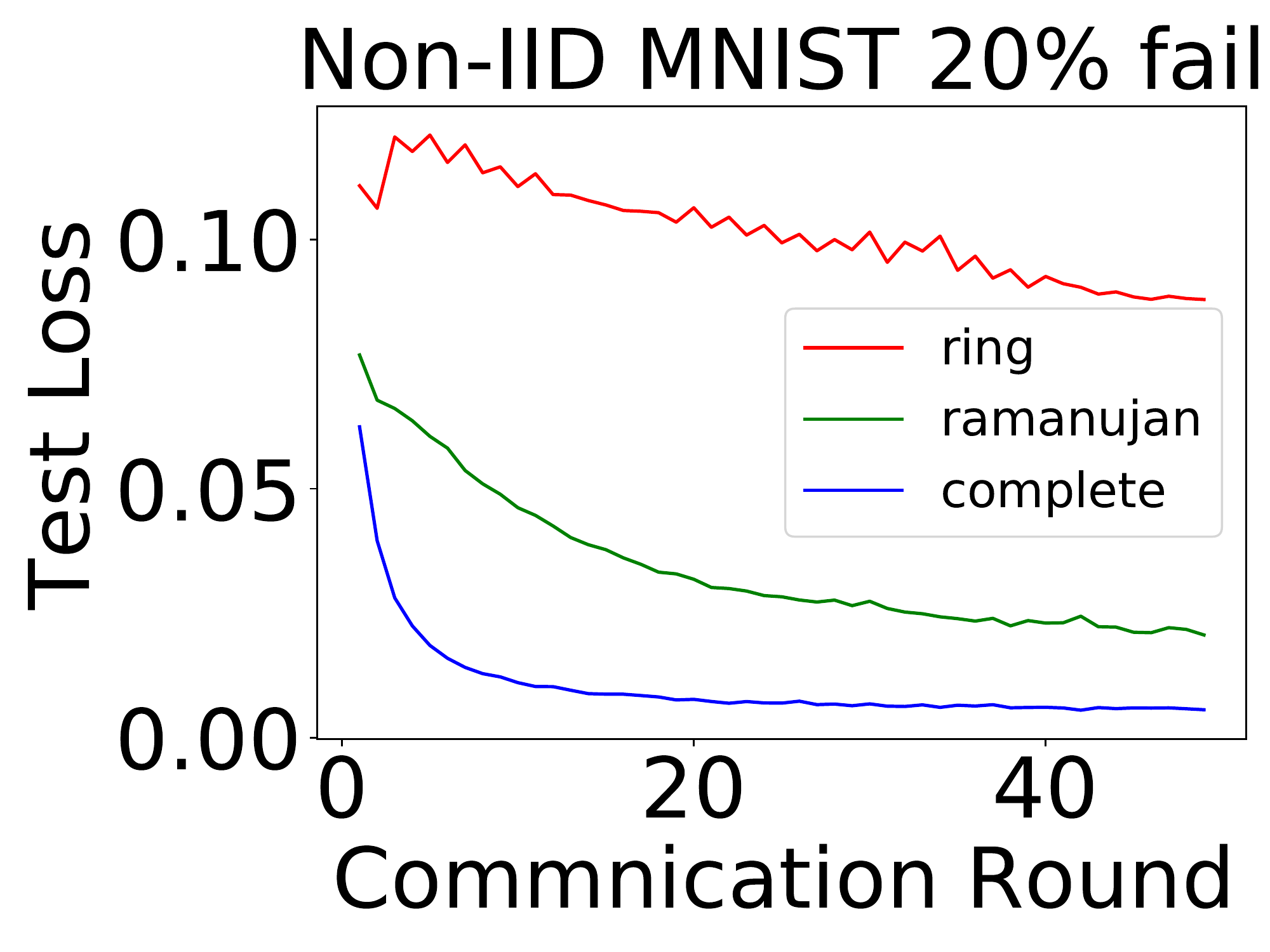}\\[-2pt]
{\footnotesize Acc 10\% failure} & {\footnotesize Loss 10\% failure} &
{\footnotesize Acc 20\% failure} & {\footnotesize Loss 20\% failure}\\
\end{tabular}\vspace{-0.4cm}
\caption{The test accuracy and test loss 
of the Ring/3-regular Ramanujan/Complete graphs on non-IID MNIST with client failures.
}
\label{fig:mnist-noniid-failure}
\end{figure}

As shown in Fig.~\ref{fig:mnist-noniid-failure}, the communication failure not only cause the loss of corresponding training samples globally, but also breaks the connection of the topology. With the weakest connectivity, the Ring graph degrades to $51.3\%$ accuracy when $20\%$ of the nodes are dropped. The clients are partitioned when multiple nodes fail in a Ring graph. The expander graph reaches $65.3\%$ of accuracy due to its high connectivity and no partition.

\paragraph{Language Modeling.}

\begin{figure}[!ht]
\centering
\begin{tabular}{cccc}
\hspace{-0.2cm}\includegraphics[width=0.23\linewidth]{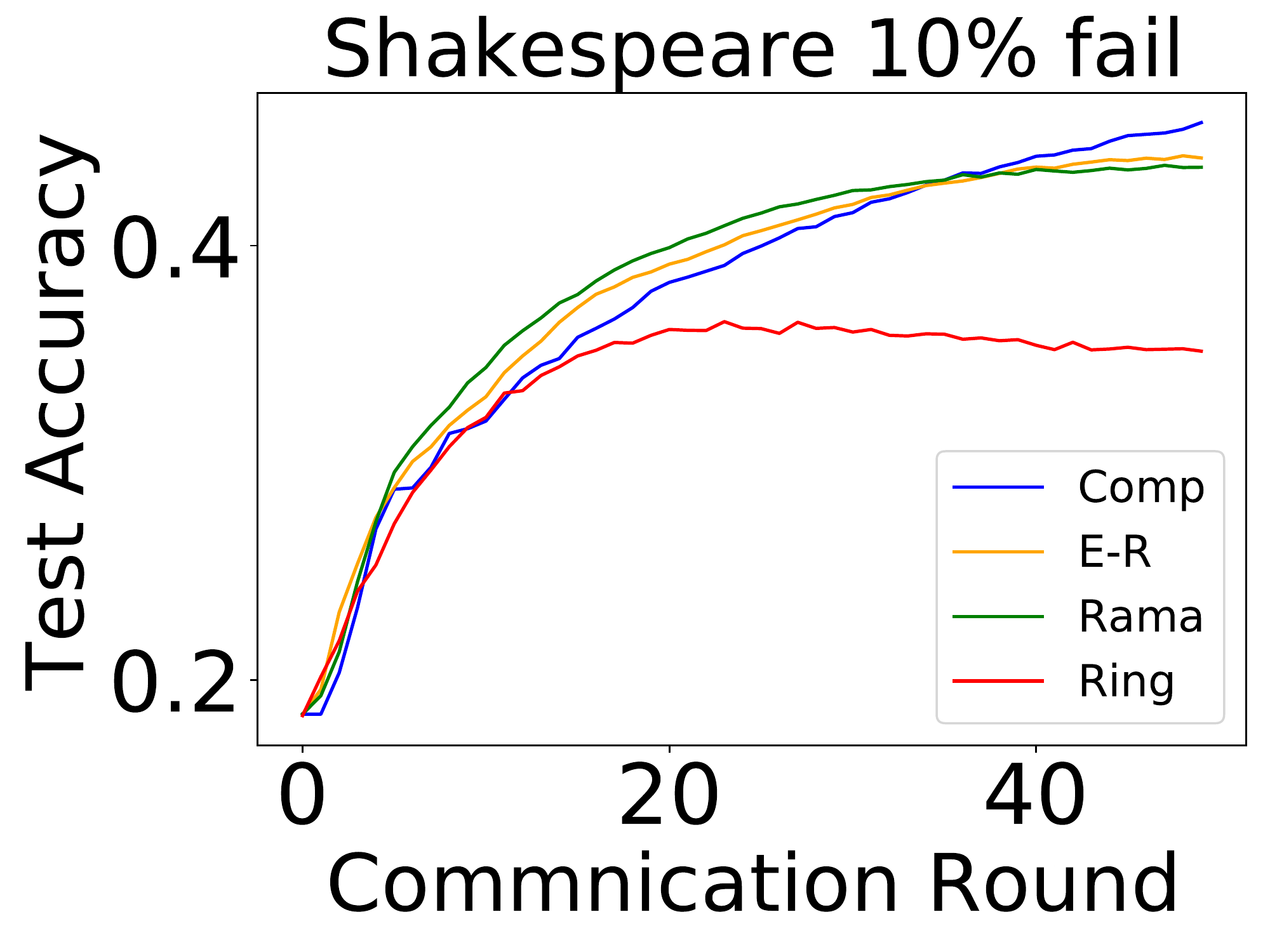}&
\hspace{-0.2cm}\includegraphics[width=0.23\linewidth]{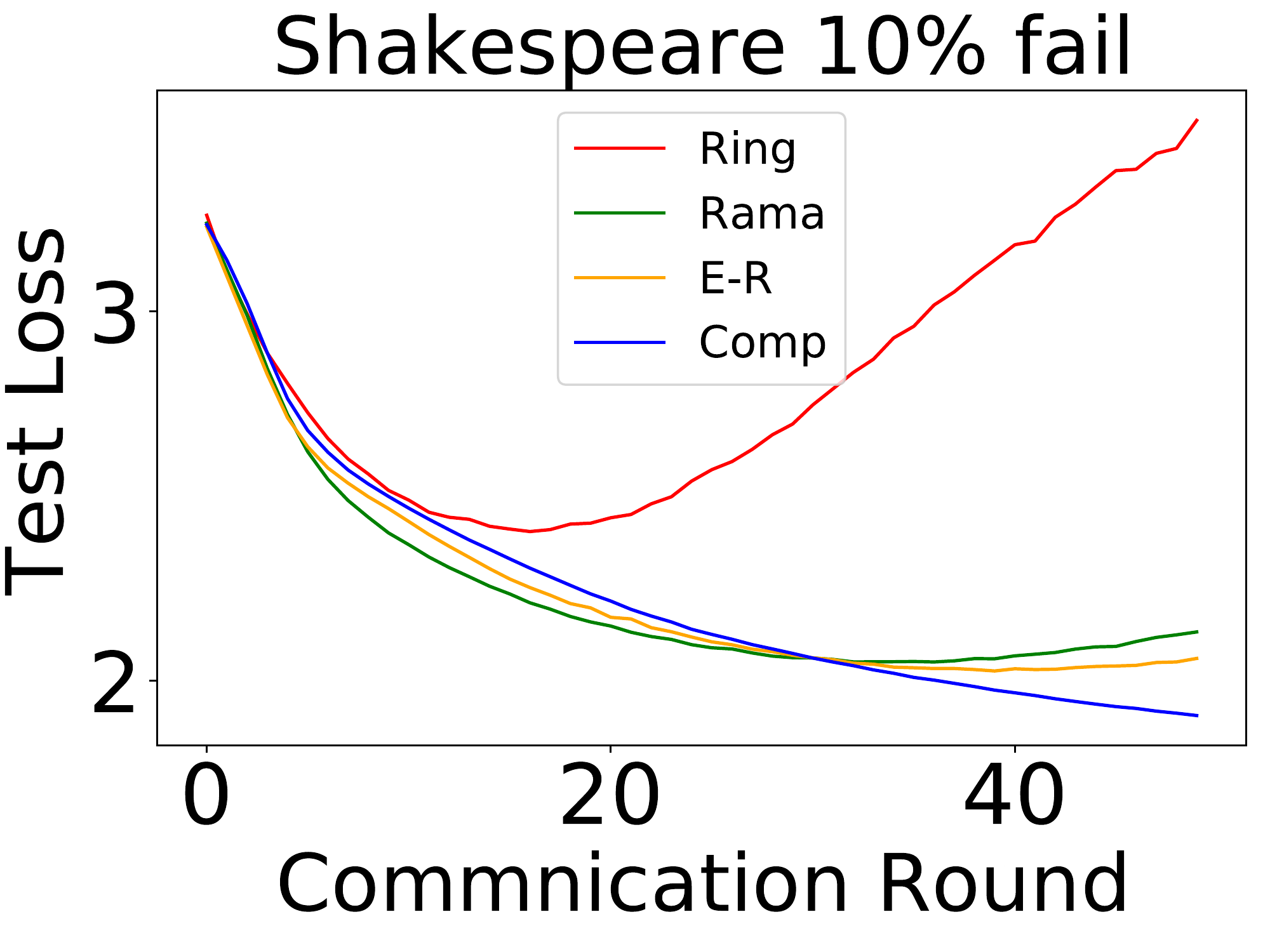}&
\hspace{-0.2cm}\includegraphics[width=0.23\linewidth]{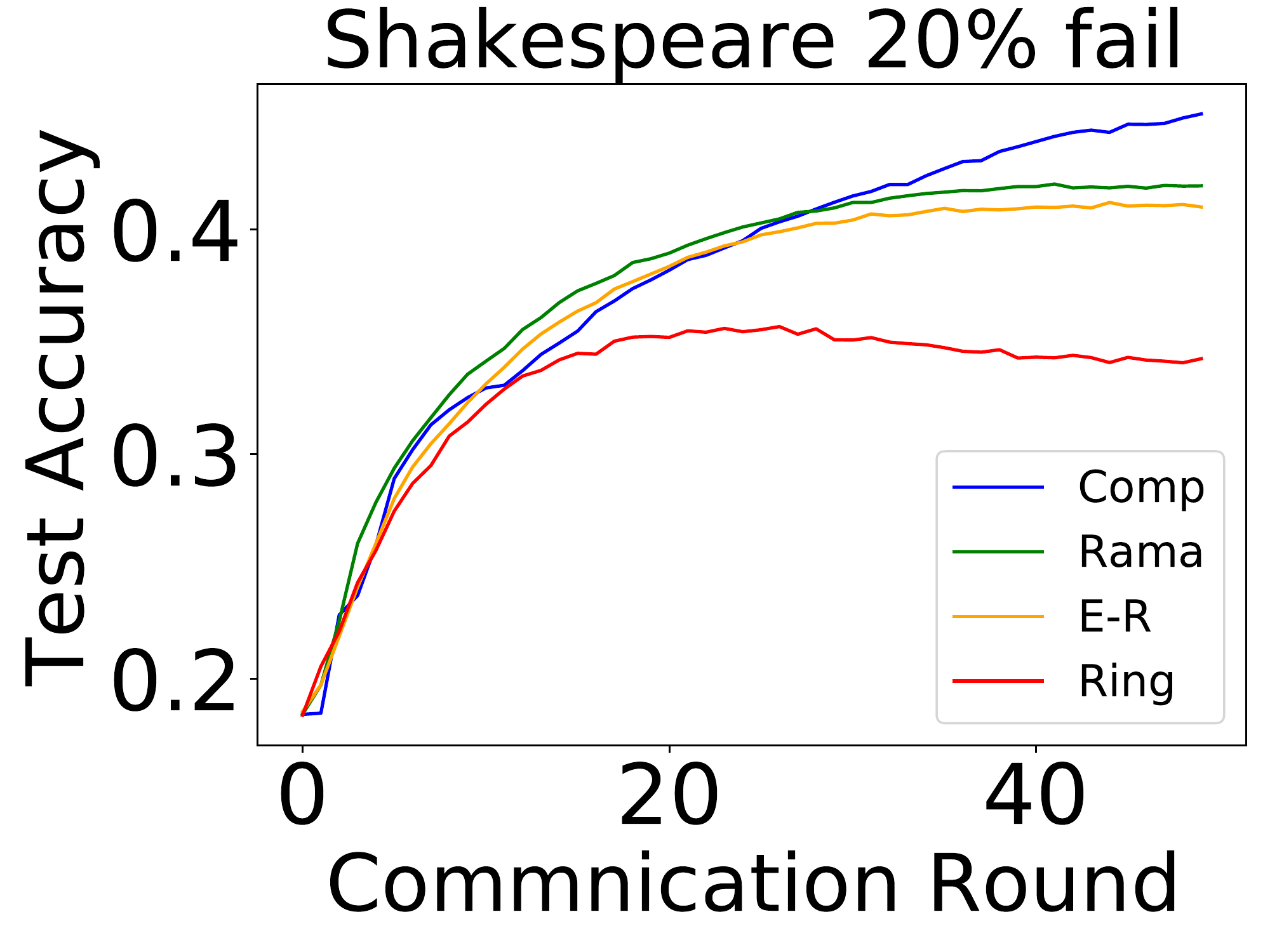}&
\hspace{-0.2cm}\includegraphics[width=0.23\linewidth]{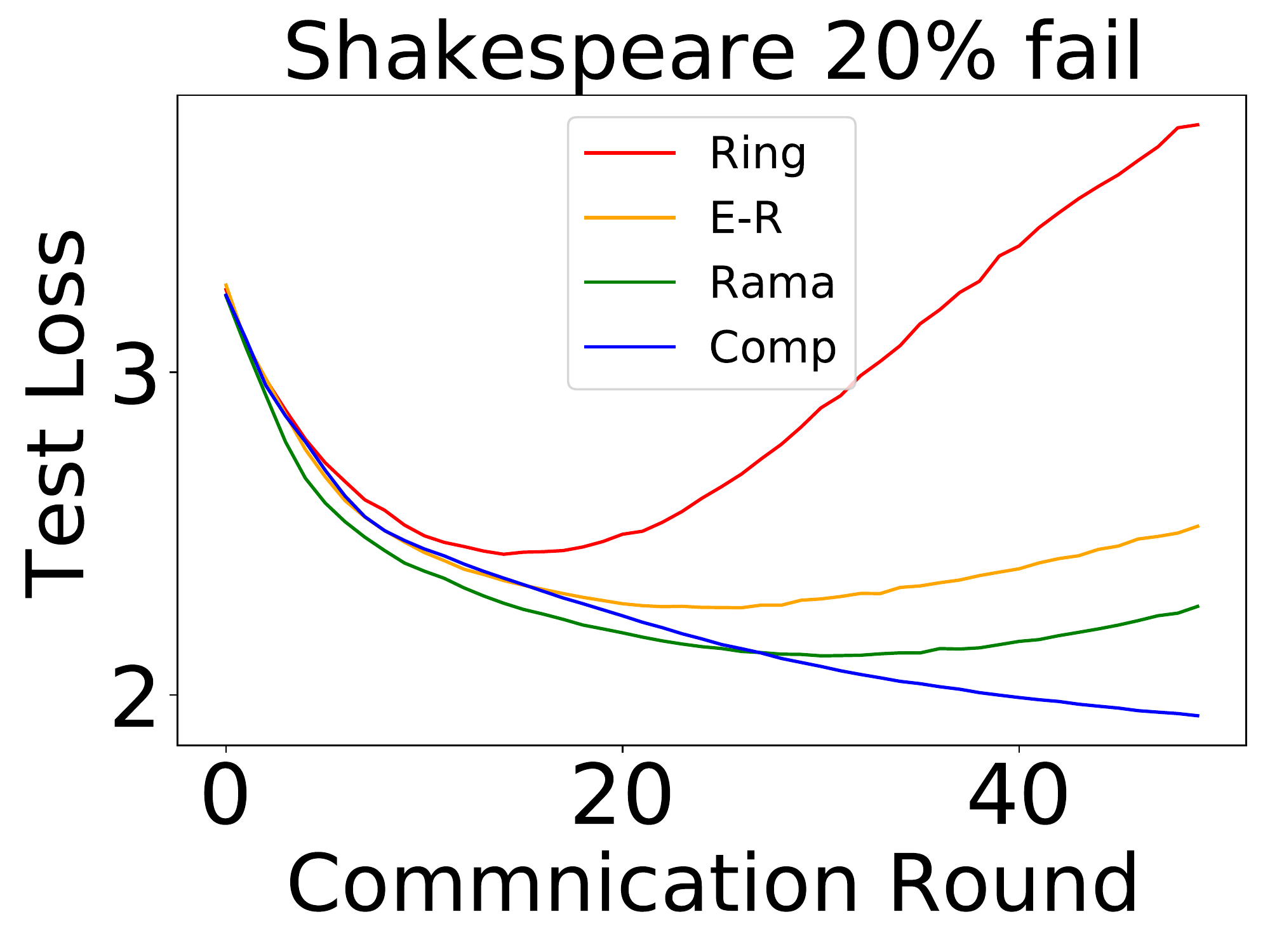}\\[-2pt]
{\footnotesize Acc 10\% failure} & {\footnotesize Loss 10\% failure} & {\footnotesize Acc 20\% failure} & {\footnotesize Loss 20\% failure}\\
\end{tabular}\vspace{-0.4cm}
\caption{The test accuracy and test loss 
the of the Ring/3-regular Ramanujan/Complete/Erd\"{o}s-R\'enyi graphs on non-IID MNIST with client failures.
}
\label{fig:shake-failure}
\end{figure}

In Fig.~\ref{fig:shake-failure}, we have a similar situation as Fig.~\ref{fig:mnist-noniid-failure}. With the weakest connectivity, the Ring graph degrades to $33.7\%$ accuracy when $20\%$ of the nodes are dropped. The clients are partitioned when multiple nodes fail in a Ring graph. The expander graph reaches $41.5\%$ of accuracy due to its merits. 
Additionally, although the Erd\"{o}s-R\'enyi graph performs slightly better than the expander graph with a $10\%$ client failures, it become worse than the expander graph with a $20\%$ client failure because of its weaker connectivity property.

\section{Concluding Remarks}\label{sec:conclusion}~
In this paper, we presented the theoretical advantages of expander graph-based overlay networks and their practical construction. We numerically verified the efficacy in accelerating training, improving generalization, and enhancing robustness to client failures of decentralized federated learning by using expander graph-based overlay networks on various benchmarks. How to establish the theoretical robustness guarantees of the expander graph-based overlay networks to the node failure is an interesting future direction.

\appendix


\section{Technical Proofs}
\begin{theorem}[General nonconvexity \cite{sun2021decentralized}, (Theorem~\ref{thm:gen-nonconvexity} restate)] \label{thm:gen-nonconvexity-restate}
Let the sequence $\{\vw_i^t\}_{t\ge 0}$ be generated by the DFedAvgM for each $i=1, 2, \ldots, N$, and suppose Assumptions~\ref{assumption:Lsmooth}-\ref{assumption:bdd-global-var} hold. Moreover, assume the constant stepsize $\eta$ satisfies $0< \eta \le 1/8LK$ and $64L^2K^2\eta^2 + 64LK\eta < 1$, where $L$ is the Lipschitz constant from Assumption~\ref{assumption:Lsmooth} and $K$ is the number of local updates before communication. Then,
{\begin{equation} 
    \min_{1 \le t \le T}\ \mbb E\|\nabla f(\bar{\vw}^t)\|^2 \le \frac{2(\bar{\vw}^1) - 2 \min f}{\gamma(K, \eta)T} + \alpha(K, \eta) + \frac{\Xi(K, \eta)}{(1-\lambda)^2},
\end{equation}}
where $T$ is the total number of communication rounds and the constants are given as 
\begin{align*}
    \gamma(K,\eta) &:= \frac{\eta(K-\beta)}{1 - \beta} - \frac{64(1 - \beta)L^2K^4\eta^3}{K-\beta} - 64LK^2\eta^2 \\
    \alpha(K, \eta) &:= \frac{\lp \frac{(1 - \beta) L^2 K^2\eta^3}{K - \beta} + L\eta^2  \rp  \lp 8K \sigma^2 + 32 K^2 \zeta^2 + \frac{64K^2 \beta^2(\sigma^2 + B^2)}{(1 - \beta)^2} \rp}{  \frac{\eta(K-\beta)}{1 - \beta} - \frac{64(1 - \beta)L^2K^4\eta^3}{K - \beta} - 64 LK^2\eta^2 } \\
    \Xi(K, \eta) &:= \lp \frac{64(1 - \beta) L^4K^4\eta^5}{K - \beta} + 64 L^3K^2\eta^4 \rp \times \\
    & \qquad \qquad \lp \frac{\lp 8K \sigma^2 + 32K^2 \zeta^2 + 32K^2B^2 + \frac{64K^2\beta^2}{(1 - \beta)^2}(\sigma^2 + B^2)  \rp}{ \lp \frac{\eta(K - \beta)}{1 - \beta} - \frac{64(1 - \beta) L^2 K^4\eta^3}{K - \beta} - 64LK^2 \eta^2  \rp}   \rp 
\end{align*}

\end{theorem}

\begin{theorem}[Uniform stability (Theorem~\ref{thm:main} restate)]
    Under Assumptions~\ref{assumption:Lsmooth}-\ref{assumption:bdd-grad}, we have that for any $T$ is the step size $\eta_t \le \frac{c}{t}$ and $c$ is small enough, then DFedAvg satisfies uniform stability with
    \[
        \epsilon_{stab} \le T^{\frac{cLK}{1 + cLK}} \lp \frac{(\sup f )K(cLK)^{\frac{1}{1 + cLK}}}{n} +  \frac{\frac{2\sigma B}{NL}}{(cLK)^{\frac{cLK}{1 + cLK}}}\rp + \frac{B(\sigma + B) \lp cK + 2C_\lambda \rp}{cLK},
    \]
    where $\sup f < \infty$ is the uniform bound on the size of the non-negative loss function.
\end{theorem}

\begingroup
\allowdisplaybreaks
\begin{proof}
Assume that each node $i$ has access to local datasets $\mcl D_i = \{(\vx_i^\ell, y_i^\ell)\}_{\ell=1}^{n_i}$ of size $n_i=n$, and denote be $\mcl D = \cup_{i=1}^N \mcl D_i$ be the set of $Nn$ datapoints over the whole graph. Assume then that the datasets $\mcl D, \tilde{\mcl D}$ differ by only one point; that is, there exists exactly one $i^\ast \in \{1, \ldots, N\}$ such $\mcl D_i$ and $\tilde{\mcl D}_i$ differ in exactly one point. Define the random variables 
\[
    \xi^{t,k}_i \sim Unif(\mcl D_i),
\]
where $\{\xi^{t,k}_i\}_{k=1}^K$ are sampled IID (with replacement). We denote the collection of random variables sampled from $\mcl D$ at all $N$ nodes in the graph as $\Xi^{(t,k)} := \{\xi^{t,k}_i\}_{i=1}^N$. Likewise, define $\tilde{\Xi}^{(t,k)} = \{\xi^{t,k}_i\}_{i=1}^N$ to be the collection of samples from $\tilde{\mcl D}$ at all $N$ nodes in the graph.

Now define $\bar{\vw}^t, \bar{\vv}^t$ to be the averages generated by DFedAvgM with training data $\mcl D, \tilde{\mcl D}$, respectively; that is, 
\[
    \bar{\vw}^t = \frac{1}{N}\sum_{i=1}^N \vw^{t,0}_i, \quad \bar{\vv}^t = \frac{1}{N}\sum_{i=1}^N  \vv^{t,0}_i.
\]
Further, define the matrices 
\[
    \mX^{(t,k)} := [\vw^{t,k}_1 \ \vw^{t,k}_2 \ \ldots \  \vw^{t,k}_N], \quad \mY^{(t,k)} := [\vv^{t,k}_1 \ \vv^{t,k}_2 \ \ldots \ \vv^{t,k}_N]
\]
and the gradient matrices
\begin{align*}
    \mG^{(t,k)}\lp \mX^{(t,k)}; \Xi^{(t,k)} \rp &:= [\nabla f_1(\vw^{t,k}_1; \xi^{t,k}_1) \  \nabla f_2(\vw^{t,k}_2; \xi^{t,k}_2) \  \ldots \ \nabla f_N(\vw^{t,k}_N; \xi^{t,k}_N)], \\
    \mG^{(t,k)}\lp \mY^{(t,k)}; \tilde{\Xi}^{(t,k)} \rp &:= [\nabla f_1(\vv^{t,k}_1; \tilde{\xi}^{t,k}_1) \ \nabla f_2(\vv^{t,k}_2; \tilde{\xi}^{t,k}_2) \ \ldots \  \nabla f_N(\vv^{t,k}_N; \tilde{\xi}^{t,k}_N)  ].
\end{align*}
We have that by definition of the DFedAvgM iterations
\begin{align*}
    &\qquad \vw^{t,k+1}_i - \vw^{t,k}_i \\
    &= -\eta_t \nabla f_i(\vw^{t,k}_i; \xi^{t,k}_i) + \theta(\vw^{t,k}_i - \vw^{t,k-1}_i) \\
    &= -\eta_t \nabla f_i(\vw^{t,k}_i; \xi^{t,k}_i) + \theta \lp -\eta_t \nabla f_i(\vw^{t,k-1}_i; \xi^{t,k-1}_i) + \theta(\vw^{t,k-1}_i - \vw^{t,k-2}_i) \rp \\
    &= \ldots \\
    &= -\eta_t \lp \sum_{s=0}^{k} \theta^{k-s}\nabla f_i(\vw^{t,s}_i; \xi^{t,s}_i) \rp
\end{align*}
and that
\begin{align} \label{eq:local-iter-sum}
    \vw^{t,k+1}_i - \vw^{t,k}_i &= -\eta_t \nabla f_i(\vw^{t,k}_i; \xi^{t,k}_i) + \theta(\vw^{t,k}_i - \vw^{t,k-1}_i) \nonumber\\
    \implies \vw^{t, K}_i - \vw^{t,0}_i &= \sum_{k=0}^{K-1} -\eta_t \nabla f_i(\vw^{t,k}_i; \xi^{t,k}_i) + \theta(\vw^{t,k}_i - \vw^{t,k-1}_i) \nonumber\\
    &= \sum_{k=0}^{K-1} -\eta_t \nabla f_i(\vw^{t,k}_i; \xi^{t,k}_i) + \theta(\vw^{t,K-1}_i - \vw^{t,0}_i)\nonumber \\
    \implies \vw^{t, K}_i - \vw^{t,0}_i &= \frac{-\eta_t}{1-\theta} \sum_{k=0}^{K-1} \nabla f_i(\vw^{t,k}_i; \xi^{t,k}_i) - \frac{\theta}{1-\theta}(\vw^{t,K}_i - \vw^{t,K-1}_i) \nonumber\\
    &= \frac{-\eta_t}{1-\theta} \sum_{k=0}^{K-1} \nabla f_i(\vw^{t,k}_i; \xi^{t,k}_i) - \frac{-\eta_t\theta}{1-\theta}\sum_{k=0}^{K-1} \sum_{s=0}^{k} \theta^{k-s} \nabla f_i(\vw^{t,s}_i; \xi^{t,s}_i) \nonumber\\
    &= \frac{-\eta_t}{1-\theta} \lp \sum_{k=0}^{K-1} \nabla f_i(\vw^{t,k}_i; \xi^{t,k}_i) - \theta\sum_{k=0}^{K-1} \nabla f_i(\vw^{t,k}_i; \xi^{t,k}_i) \sum_{s=0}^{K-k-1} \theta^s \rp \nonumber\\
    &= \frac{-\eta_t}{1-\theta} \sum_{k=0}^{K-1} \lp 1 - \frac{\theta - \theta^{K-k+1}}{1-\theta} \rp \nabla f_i(\vw^{t,k}_i; \xi^{t,k}_i) \nonumber\\
    &= \frac{\eta_t}{(1-\theta)^2} \sum_{k=0}^{K-1} (1 - 2\theta + \theta^{K-k+1})\nabla f_i(\vw^{t,k}_i; \xi^{t,k}_i).
\end{align}
Then by \ref{eq:local-iter-sum} we can write
\begin{align*}
    \mX^{(t,K)} - \mX^{(t,0)} &= \frac{\eta_t }{(1-\theta)^2} \sum_{k=0}^{K-1} p_k(\theta) \mG^{(t,k)}(\mX^{(t,k)}; \Xi^{(t,k)}),  \\
    \mY^{(t,K)} - \mY^{(t,0)} &= \frac{\eta_t }{(1-\theta)^2} \sum_{k=0}^{K-1} p_k(\theta)\mG^{(t,k)}(\mY^{(t,k)}; \tilde{\Xi}^{(t,k)}),
\end{align*}
where we have defined $p_k(\theta) = 1 - 2\theta + \theta^{K - k + 1}$.

Letting $\mathbbm{1} \in \mbb R^N$ denote the vector of all ones, then the mixing matrix $W$ satisfies $W \frac{\mathbbm{1}}{N} = \frac{\mathbbm{1}}{N}$. Then we have with probability $\lp\frac{n-1}{n}\rp^K$ the random variables $\{\Xi^{(t,k)}\}_{k=1}^K = \{\tilde{\Xi}^{(t,k)}\}_{k=1}^K$ are exactly the same:
{\footnotesize
\begin{align} \label{eq:xminusy}
    &\qquad \bar{\vw}^{t+1} - \bar{\vv}^{t+1} \\
    &= \mX^{(t,K)}W\frac{\mathbbm{1}}{N} - \mY^{(t,K)}W\frac{\mathbbm{1}}{N} \nonumber \\
    &= \lp \mX^{(t,0)} + \lp \mX^{(t,K)} - \mX^{(t,0)}\rp \rp \frac{\mathbbm{1}}{N} - \lp \mY^{(t,0)} + \lp \mY^{(t,K)} - \mY^{(t,0)}\rp \rp \frac{\mathbbm{1}}{N} \nonumber \\
    &= \lp \mX^{(t,0)} - \mY^{(t,0)} + \frac{\eta_t }{(1-\theta)^2}\lp \sum_{k=0}^{K-1} p_k(\theta)\mG^{(t,k)}(\mX^{(t,k)}; \Xi^{(t,k)}) -  \sum_{k=0}^{K-1} p_k(\theta)\mG^{(t,k)}(\mY^{(t,k)}; \Xi^{(t,k)})\rp \rp \frac{\mathbbm{1}}{N} \nonumber\\
    &= \lp\lp \mX^{(t,0)} - \mY^{(t,0)}\rp (\mI - \mP) + \lp \mX^{(t,0)} - \mY^{(t,0)}\rp P \rp\frac{\mathbbm{1}}{N} \nonumber\\
    & \quad + \frac{\eta_t }{(1-\theta)^2}\lp \sum_{k=0}^{K-1} p_k(\theta) \left[ \mG^{(t,k)}(\mX^{(t,k)}; \Xi^{(t,k)})  - \mG^{(t,k)}(\bar{\vw}^t\mathbbm{1}^T; \Xi^{(t,k)}) + \mG^{(t,k)}(\bar{\vw}^t\mathbbm{1}^T; \Xi^{(t,k)})\right] \rp   \frac{\mathbbm{1}}{N} \nonumber \\
    & \qquad  -\frac{\eta_t }{(1-\theta)^2}\lp \sum_{k=0}^{K-1} p_k(\theta) \left[ \mG^{(t,k)}(\mY^{(t,k)}; \Xi^{(t,k)})  - \mG^{(t,k)}(\bar{\vv}^t\mathbbm{1}^T; \Xi^{(t,k)}) + \mG^{(t,k)}(\bar{\vv}^t\mathbbm{1}^T; \Xi^{(t,k)})\right] \rp \frac{\mathbbm{1}}{N} \nonumber \\
    &= \underbrace{\frac{\eta_t }{N(1-\theta)^2} \sum_{i=1}^N \sum_{k=0}^{K-1}p_k(\theta)\left[ \lp \nabla f_i(\bar{\vw}^t; \xi^{t,k}_i) - \nabla f_i(\vw^{t,k}_i; \xi^{t,k}_i)\rp - \lp \nabla f_i(\bar{\vv}^t; \xi^{t,k}_i) - \nabla f_i(\vv^{t,k}_i; \xi^{t,k}_i)\rp\right]}_{=: A_1} \nonumber\\
    &\qquad + \frac{1}{N} \sum_{i=1}^N \lp \bar{\vw}^t - \frac{\eta_t }{(1-\theta)^2} \sum_{k=0}^{K-1} p_k(\theta) \nabla f_i (\bar{\vw}^t; \xi^{t,k}_i)\rp - \lp \bar{\vv}^t - \frac{\eta_t }{(1-\theta)^2} \sum_{k=0}^{K-1} p_k(\theta) \nabla f_i (\bar{\vv}^t; \xi^{t,k}_i)\rp,
\end{align}}
where we note that $(\mI - \mP)\frac{\mathbbm{1}}{N} = 0$. Now, we have that since $\theta \in [0, 1)$, then $|p_k(\theta)| = p_k(\theta) \le p_{K-1}(\theta) = (1 - \theta)^2$ for each $k=0, 1, \ldots, K-1$. This means we can calculate
\begin{align*}
    &\qquad \|A_1\|\\
    &\le \frac{\eta_t}{N} \sum_{i=1}^N \sum_{k=0}^{K-1} \| \nabla f_i(\bar{\vw}^t; \xi^{t,k}_i) - \nabla f_i(\vw^{t,k}_i; \xi^{t,k}_i)\|  +  \| \nabla f_i(\bar{\vv}^t; \xi^{t,k}_i) - \nabla f_i(\vv^{t,0}_i; \xi^{t,k}_i)\| \\
    &\le  \frac{\eta_tL}{N} \sum_{i=1}^N \sum_{k=0}^{K-1} \| \bar{\vw}^t -\vw^{t,k}_i\|  +  \| \bar{\vv}^t- \vv^{t,k}_i\| \\
    &\le \frac{\eta_tL}{N} \sum_{k=0}^{K-1} \sum_{i=1}^N  \lp\| \bar{\vw}^t -\vw^{t,0}_i\|  +  \| \bar{\vv}^t- \vv^{t,0}_i\| \rp \\
    &\qquad + \frac{\eta_tL}{N}\sum_{k=1}^{K-1} \sum_{i=1}^N  \lp \|\vw^{t,0}_i -\vw^{t,k}_i\|  +  \| \vv^{t,0}_i- \vv^{t,k}_i\| \rp \\
    &\le \frac{\eta_tL}{\sqrt{N}} \sum_{k=0}^{K-1} \lp\| \mX^{(t,0)}(\mI - \mP)\|_F +  \| \mY^{(t,0)}(\mI - \mP)\|_F \rp \\
    &\qquad \qquad  + \frac{\eta_t^2L}{N(1 - \theta)^2}  \sum_{i=1}^N \sum_{k=1}^{K-1}  \sum_{s=0}^{k-1}  p_s(\theta) \lp \|\nabla f_i(\vw^{t,s}_i; \xi^{t,s}_i)\|  +  \| \nabla f_i(\vv^{t,s}_i; \xi^{t,s}_i)\| \rp \\
    &\le \frac{\eta_tL}{\sqrt{N}} \sum_{k=0}^{K-1} \lp\| \mX^{(t,0)}(\mI - \mP)\|_F +  \| \mY^{(t,0)}(\mI - \mP)\|_F \rp \\
    &\qquad \qquad  + \frac{\eta_t^2L}{N}  \sum_{i=1}^N \sum_{k=1}^{K-1}  \sum_{s=0}^{k-1} \lp \|\nabla f_i(\vw^{t,s}_i; \xi^{t,s}_i)\|  +  \| \nabla f_i(\vv^{t,s}_i; \xi^{t,s}_i)\| \rp.
\end{align*}

Then, we have
\begin{align*}
    \frac{1}{N}\sum_{i=1}^N \|\nabla f_i(\vw^{t,s}_i; \xi^{t,s}_i)\| &\le \frac{1}{N} \sum_{i=1}^N \|\nabla f_i(\vw^{t,s}_i; \xi^{t,s}_i) - \nabla f_i(\vw^{t,s}_i)\| + \|\nabla f_i(\vw^{t,s}_i)\| \\
    &\le \frac{1}{\sqrt{N}} \lp \sum_{i=1}^N \|\nabla f_i(\vw^{t,s}_i; \xi^{t,s}_i) - \nabla f_i(\vw^{t,s}_i)\|^2 \rp^{\frac{1}{2}} + B \\
    \implies \mbb E \frac{1}{N}\sum_{i=1}^N \|\nabla f_i(\vw^{t,s}_i; \xi^{t,s}_i)\| &\le \frac{1}{\sqrt{N}} \lp \sum_{i=1}^N \mbb E \|\nabla f_i(\vw^{t,s}_i; \xi^{t,s}_i) - \nabla f_i(\vw^{t,s}_i)\|^2 \rp^{\frac{1}{2}} + B \\
    &\le \frac{1}{\sqrt{N}} \lp N \sigma^2\rp^{\frac{1}{2}} + B = \sigma + B,
\end{align*}
so that with applying the \Cref{lemma:for-stability} 
\begin{align*}
    \|A_1\| &\le 2\eta_tLK(\sigma + B)\lp \sum_{j=1}^{t} \eta_{t-j}\lambda^j\rp + 2\eta_t^2L (\sigma + B) 
    \sum_{k=1}^{K-1} k \\
    &= 2\eta_tLK(\sigma + B)\lp \sum_{j=1}^{t} \eta_{t-j}\lambda^j\rp + \eta_t^2L(\sigma + B) K(K-1) \\
    &\le \eta_t LK(\sigma + B) \lp 2 \sum_{j=1}^{t-1} \eta_{t-j}\lambda^j + \eta_t K\rp.
\end{align*}
Now, noticing that with each $f_i$ being L-smooth, we can calculate 
\begin{align*}
    &\left\| \bar{\vw}^t - \frac{\eta_t }{(1-\theta)^2} \sum_{k=0}^{K-1} p_k(\theta) \nabla f_i (\bar{\vw}^t; \xi^{t,k}_i) - \lp \bar{\vv}^t - \frac{\eta_t }{(1-\theta)^2} \sum_{k=0}^{K-1} p_k(\theta) \nabla f_i (\bar{\vv}^t; \xi^{t,k}_i)\rp\right\| \\
    & \qquad \qquad \qquad \qquad\le \|\bar{\vw}^t - \bar{\vv}^t\| + \eta_t \sum_{k=0}^{K-1}\left\| \nabla f_i (\bar{\vw}^t; \xi^{t,k}_i) - \nabla f_i (\bar{\vv}^t; \xi^{t,k}_i) \right\|  \\
    & \qquad \qquad \qquad \qquad \le ( 1+ \eta_t LK) \|\bar{\vw}^t - \bar{\vv}^t\|.
\end{align*}
Plugging everything into (\ref{eq:xminusy})
\begin{align*}
    \mbb E \|\bar{\vw}^{t+1} - \bar{\vv}^{t+1}\| &\le (1 + \eta_t LK) \|\bar{\vw}^t - \bar{\vv}^t\| + \eta_t LK(\sigma + B) \lp 2 \sum_{j=1}^{t-1} \eta_{t-j}\lambda^j + \eta_t K\rp.
\end{align*}

Now, with probability $1 - \lp \frac{n-1}{n}\rp^K$, we have that the random variables $\{\tilde{\Xi}^{(t,k)}\}_{k=1}^K$ might be different from  $\{\Xi^{(t,k)}\}_{k=1}^K$ {\it in the draws from node $i^\ast$}. We calculate, similarly to the previous case,
{\footnotesize \begin{align*}
    &\qquad \bar{\vw}^{t+1} - \bar{\vv}^{t+1}\\
    &= \frac{\eta_t }{N(1-\theta)^2} \sum_{i=1}^N \sum_{k=0}^{K-1}p_k(\theta)\left[ \lp \nabla f_i(\bar{\vw}^t; \xi^{t,k}_i) - \nabla f_i(\vw^{t,k}_i; \xi^{t,k}_i)\rp - \lp \nabla f_i(\bar{\vv}^t; \tilde{\xi}^{t,k}_i) - \nabla f_i(\vv^{t,k}_i; \tilde{\xi}^{t,k}_i)\rp\right] \nonumber\\
    &\qquad + \frac{1}{N} \sum_{i=1}^N \lp \bar{\vw}^t - \frac{\eta_t }{(1-\theta)^2} \sum_{k=0}^{K-1} p_k(\theta) \nabla f_i (\bar{\vw}^t; \xi^{t,k}_i)\rp - \lp \bar{\vv}^t - \frac{\eta_t }{(1-\theta)^2} \sum_{k=0}^{K-1} p_k(\theta) \nabla f_i (\bar{\vv}^t; \tilde{\xi}^{t,k}_i)\rp
\end{align*}}
which allows us to conclude by performing the same calculations we did on $A_1$
{\footnotesize \begin{align}
    \implies &\mbb E \|\bar{\vw}^{t+1} - \bar{\vv}^{t+1} \| \le \eta_t LK(\sigma + B) \lp 2 \sum_{j=1}^{t-1} \eta_{t-j}\lambda^j + \eta_t K\rp \\
    &+ \mbb E \left\|\underbrace{\frac{1}{N} \sum_{i=1}^N \lp \bar{\vw}^t - \frac{\eta_t }{(1-\theta)^2} \sum_{k=0}^{K-1} p_k(\theta) \nabla f_i (\bar{\vw}^t; \xi^{t,k}_i)\rp - \lp \bar{\vv}^t - \frac{\eta_t }{(1-\theta)^2} \sum_{k=0}^{K-1} p_k(\theta) \nabla f_i (\bar{\vv}^t; \tilde{\xi}^{t,k}_i)\rp}_{=:A_2} \right\|.
\end{align}}

Now, turning our attention to the term $A_2$, we can calculate 
{\footnotesize\begin{align*}
    A_2 &= \frac{1}{N}\sum_{i=1}^N \left[ \bar{\vw}^t - \bar{\vv}^t + \frac{\eta_t }{(1-\theta)^2} \sum_{k=0}^{K-1} p_k(\theta)\lp \nabla f_i (\bar{\vw}^t; \xi^{t,k}_i)  - \nabla f_i (\bar{\vv}^t; \tilde{\xi}^{t,k}_i)\rp \right] \\
    &= \bar{\vw}^t - \bar{\vv}^t + \frac{\eta_t}{N(1-\theta)^2} \sum_{i \not= i^\ast} \sum_{k=0}^{K-1} p_k(\theta)\lp \nabla f_i (\bar{\vw}^t; \xi^{t,k}_i)  - \nabla f_i (\bar{\vv}^t; \xi^{t,k}_i)\rp \\
    &\qquad\qquad + \frac{\eta_t}{N(1-\theta)^2} \sum_{k=0}^{K-1} p_k(\theta)\lp \nabla f_{i^\ast} (\bar{\vw}^t; \xi^{t,k}_{i^\ast})  - \nabla f_{i^\ast} (\bar{\vv}^t; \tilde{\xi}^{t,k}_{i^\ast})\rp \\
    \implies \|A_2\| &\le \|\bar{\vw}^t - \bar{\vv}^t\| + \frac{\eta_t LK (m-1)}{N} \|\bar{\vw}^t - \bar{\vv}^t\| +\frac{\eta_t}{N} \sum_{k=0}^{K-1} \|\nabla f_{i^\ast} (\bar{\vw}^t; \xi^{t,k}_{i^\ast})  - \nabla f_{i^\ast} (\bar{\vw}^t)\|\\
    &\qquad \qquad + \frac{\eta_t}{N} \sum_{k=0}^{K-1} \|\nabla f_{i^\ast} (\bar{\vw}^t)  - \nabla f_{i^\ast} (\bar{\vv}^t)\| + \|\nabla f_{i^\ast} (\bar{\vv}^t)  - \nabla f_{i^\ast} (\bar{\vv}^t; \tilde{\xi}^{t,k}_{i^\ast})\| \\
    \implies \mbb E\|A_2\| &\le \lp 1 + \frac{\eta_t LK (m-1)}{N}\rp \mbb E \|\bar{\vw}^t - \bar{\vv}^t\| + \frac{\eta_t K}{N}\lp 2\sigma +  L \mbb E \|\bar{\vw}^t - \bar{\vv}^t\| \rp \\
    &= \lp 1 + \eta_t LK \rp \mbb E \|\bar{\vw}^t - \bar{\vv}^t\| + \frac{2\eta_t \sigma K}{N},
\end{align*}}
where in the second to last line we have used the fact that by Jensen's inequality
\[
    \mbb E \lp \|\bz\|^2 \rp^{\frac{1}{2}} \le \lp \mbb E \|\bz\|^2 \rp^{\frac{1}{2}}
\]
combined with the assumption of bounded variance of stochastic gradients.

Recalling the definition $\delta_t := \|\bar{\vw}^t - \bar{\vv}^t\|$, then we can combine both cases to obtain
\begin{align*}
    &\qquad \mbb E (\delta_{t+1} | \delta_{t_0} = 0)\\
    &\le \lp \frac{n-1}{n}\rp^K \lp (1 + \eta_t LK) \mbb E (\delta_{t} | \delta_{t_0} = 0) + \eta_t LK(\sigma + B) \lp 2 \sum_{j=1}^{t-1} \eta_{t-j}\lambda^j + \eta_t K\rp \rp \\
    & + \lp 1 - \lp\frac{n-1}{n}\rp^K  \rp  \lp \eta_t LK(\sigma + B) \lp 2 \sum_{j=1}^{t-1} \eta_{t-j}\lambda^j + \eta_t (K+1)\rp \rp \\ 
    &\qquad + \lp 1 - \lp\frac{n-1}{n}\rp^K  \rp  \lp \lp 1 + \eta_t LK \rp \mbb E (\delta_{t} | \delta_{t_0} = 0) + \frac{2\eta_t \sigma K}{N}  \rp  \\
    &=\lp 1 + \eta_t LK \rp \mbb E (\delta_{t} | \delta_{t_0} = 0) + \eta_t LK(\sigma + B) \lp 2 \sum_{j=1}^{t-1} \eta_{t-j}\lambda^j + \eta_t K\rp \\
    &\qquad + \lp 1 - \lp\frac{n-1}{n}\rp^K  \rp \frac{2\eta_t \sigma K}{N}.
\end{align*}

With a similar result to bound the sum $\sum_{j=1}^{t-1} \eta_{t-j} \lambda^j$ to that of \cite{sun2021decentralized}, if we set
\[
    \eta_t \le \frac{c}{t},
\]
then we should be able to calculate

\begin{align*}
    \mbb E (\delta_{t+1} | \delta_{t_0} = 0) &\le \lp 1 + \frac{cLK}{t} \rp \mbb E (\delta_{t} | \delta_{t_0} = 0) +  \frac{cLK}{t}(\sigma + B) \lp 2\frac{C_\lambda}{t} +  \frac{cK}{t} \rp \\
    &\qquad + \lp 1 - \lp\frac{n-1}{n}\rp^K  \rp \frac{2c \sigma K}{Nt} \\
    &\le  \lp 1 + \frac{cLK}{t} \rp \mbb E (\delta_{t} | \delta_{t_0} = 0)  + \underbrace{ \frac{2cK\sigma }{N} }_{=: C_1} \frac{1}{t} +  \underbrace{cLK(\sigma + B) \lp cK + 2C_\lambda \rp}_{=: C_2}\frac{1}{t^2}\\
    &\le \exp\lp \frac{cLK}{t} \rp\mbb E (\delta_{t} | \delta_{t_0} = 0) + \frac{C_1}{t} + \frac{C_2}{t^2}.
\end{align*}
Unraveling this recursion, we obtain
\begin{align*}
    \mbb E (\delta_{T} | \delta_{t_0} = 0) &\le \sum_{t=t_0+1}^T   \exp\lp cLK \sum_{k=t+1}^T \frac{1}{k} \rp  \lp \frac{C_1}{t} + \frac{C_2}{t^2}  \rp \\
    &\le \sum_{t=t_0+1}^T   \exp\lp cLK \ln\frac{T}{t} \rp  \lp \frac{C_1}{t} + \frac{C_2}{t^2}  \rp \\
    &= T^{cLK} \lp \sum_{t=t_0+1}^T \frac{1}{t^{cLK+1}}\lp C_1 + \frac{C_2}{t} \rp \rp  \\
    &\le T^{cLK}\lp \frac{C_1}{cLK t_0^{cLK}}  + \frac{C_2}{(cLK + 1)t_0^{cLK + 1}}\rp \\
    &\le \lp\frac{T}{t_0}\rp^{cLK}\frac{1}{cLK}\lp C_1 + \frac{C_2}{t_0}\rp.
\end{align*}
Plugging in the definitions of $C_1, C_2$ we get
\begin{align*}
    \mbb E (\delta_{T} | \delta_{t_0} = 0) &\le \lp\frac{T}{t_0}\rp^{cLK}\frac{1}{cLK}\left[ \frac{2cK\sigma }{N}  + \frac{cLK(\sigma + B) \lp cK + 2C_\lambda \rp}{t_0}\right] \\
    &= \lp\frac{T}{t_0}\rp^{cLK}\left[ \frac{2\sigma}{NL} + \frac{(\sigma + B) \lp cK + 2C_\lambda \rp}{t_0} \right],
\end{align*}
which gives by \Cref{lemma:hardt}
\begin{align*}
    \mbb E | f(\bar{\vw}^T; \Xi) - f(\bar{\vv}^T; \Xi)| &\le t_0 (\sup f ) \lp 1 - \lp \frac{n-1}{n}\rp^K \rp \\
    &\qquad + B \lp\frac{T}{t_0}\rp^{cLK}\left[ \frac{2\sigma}{NL} + \frac{(\sigma + B) \lp cK + 2C_\lambda \rp}{t_0} \right] \\
    &\le \frac{t_0K}{n} (\sup f ) + \lp\frac{T}{t_0}\rp^{cLK}\left[ \frac{2\sigma B}{NL} + \frac{B(\sigma + B) \lp cK + 2C_\lambda \rp}{t_0} \right].
\end{align*}

The right hand side is approximately minimized if we choose 
$$t_0 = T^{\frac{cLK}{1 + cLK}} (cLK)^{\frac{1}{1 + cLK}},$$
which we can ensure is less than $n$ for $c$ sufficiently small. We then can calculate
\begin{align*}
    \mbb E | f(\bar{\vw}^T; \Xi) - f(\bar{\vv}^T; \Xi)| &\le  \frac{(\sup f )K(cLK)^{\frac{1}{1 + cLK}}}{n} T^{\frac{cLK}{1 + cLK}}   + \frac{\frac{2\sigma B}{NL}}{(cLK)^{\frac{cLK}{1 + cLK}}} T^{\frac{cLK}{1 + cLK}}  \\
    &\qquad \qquad +  \frac{B(\sigma + B) \lp cK + 2C_\lambda \rp T^{cLK}}{\lp(cLK)^{\frac{1}{1 + cLK}}T^{\frac{cLK}{1+cLK}}\rp^{cLK+1}} \\
    &= T^{\frac{cLK}{1 + cLK}} \lp \frac{(\sup f )K(cLK)^{\frac{1}{1 + cLK}}}{n} +  \frac{\frac{2\sigma B}{NL}}{(cLK)^{\frac{cLK}{1 + cLK}}}\rp \\
    &\qquad + \frac{B(\sigma + B) \lp cK + 2C_\lambda \rp}{cLK},
\end{align*}
as desired.
\end{proof}
\endgroup

\begin{lemma}\label{lemma:for-stability}
    Under Assumptions~\ref{assumption:Lsmooth}-\ref{assumption:bdd-grad} and on the mixing matrix $W$, we have that 
    \[
         \| \mX^{(t,0)}(\mI - \mP)\|_F \le K \sqrt{N}(\sigma + B) \sum_{j=1}^{t-1} \eta_{t-j}\lambda^j,
    \]
    where $\sigma, B$ and $\lambda$ are constants from our assumptions and $K$ is the number of local updates performed before aggregation via the graph topology.
\end{lemma}

\begin{proof}
    Let the vector $\bar{\vw}^{t,k} = \sum_{i=1}^N \vw^{t,k}_i / N$ be the average parameter vector during intermediate, local updates. Then, let the matrix of ``true gradients'' of the global objective function $f$  be
    \[
        \nabla \mF(\bX^{(t,k)}) := [\nabla f(\bar{\vw}^{t,k}) \ \nabla f(\bar{\vw}^{t,k}) \ \ldots \nabla f(\bar{\vw}^{t,k})],
    \]
    obtained by horizontally concatenating the true gradient vector $\nabla f(\bar{\vw}^{t,k})$.
    Recalling that $p_k(\theta) \le (1-\theta)^2$ for $k=0, \ldots, K-1$, we have
    {\small \begin{align*}
        &\quad \| \mX^{(t,0)}(\mI - \mP)\|_F\\
        &=\left\|  \sum_{j=1}^{t} \frac{\eta_{t-j}}{(1-\theta)^2} \sum_{k=0}^{K-1} p_k(\theta) \mG^{(t,k)}(\mX^{(t,k)}; \Xi^{(t,k)}) \lp W^j - P\rp \right\|_F \\
        &\le \sum_{j=1}^{t} \eta_{t-j}\left\| W^j - P \right\|_F \sum_{k=0}^{K-1} \left\|  \mG^{(t-j,k)}(\mX^{(t-j,k)}; \Xi^{(t,k)}) - \nabla \mF(\mX^{(t-j,k)}) + \nabla \mF(\mX^{(t-j,k)}) \right\|_F  \\
        &\le \sum_{j=1}^{t} \eta_{t-j} \lambda^j \sum_{k=0}^{K-1} \lp \sum_{i=1}^N \|\nabla f_i(\vw^{t-j,k}_i; \xi^{t-j,k}_i) - \nabla f_i(\vw^{t-j,k}_i)\|^2 \rp^{\frac{1}{2}} 
         + \lp \sum_{i=1}^N \|\nabla f_i(\vw^{t-j,k}_i)\|^2 \rp^{\frac{1}{2}} \\
        &\le \sum_{j=1}^{t} \eta_{t-j} \lambda^j \sum_{k=0}^{K-1} \left[ \lp \sum_{i=1}^N \|\nabla f_i(\vw^{t-j,k}_i; \xi^{t-j,k}_i) - \nabla f_i(\vw^{t-j,k}_i)\|^2 \rp^{\frac{1}{2}} + B \sqrt{N} \right],
        \end{align*}
        \begin{align*}
        \implies &\qquad \mbb E \| \mX^{(t,0)}(\mI - \mP)\|_F\\
        &\le \sum_{j=1}^{t} \eta_{t-j} \lambda^j \sum_{k=0}^{K-1} \left[ \lp \sum_{i=1}^N \mbb E\|\nabla f_i(\vw^{t-j,k}_i; \xi^{t-j,k}_i) - \nabla f_i(\vw^{t-j,k}_i)\|^2 \rp^{\frac{1}{2}} + B \sqrt{N} \right] \\
        &\le \sum_{j=1}^{t} \eta_{t-j} \lambda^j \sum_{k=0}^{K-1} \left[ \sigma \sqrt{N} + B  \sqrt{N}\right] \\
        &= K \sqrt{N}(\sigma + B)\sum_{j=1}^{t} \eta_{t-j}\lambda^j,
    \end{align*}}
    where in the third line from the bottom we have used Jensen's inequality, since the square root function is concave.
\end{proof}

\begin{lemma} \label{lemma:hardt}
    Assume that the loss function $f(\cdot; \Xi)$ is nonnegative and $B$-Lipschitz for all $\Xi$. Let $\mcl D$, $\tilde{\mcl D}$ be two samples of size $Nn$ differing in only a single example. Let $\bar{\vw}^T, \bar{\vv}^T$ denote the output of DFedAvgM after $T$ steps with the dataset samples $\mcl D$ and $\tilde{\mcl D}$, respectively. Then, for every $\Xi$ and every $t_0 \in \{0,1,\ldots, n\}$, under the random selection rule, we have
    \[
        \mbb E | f(\bar{\vw}^T; \Xi) - f(\bar{\vv}^T; \Xi)| \le t_0 (\sup f ) \lp 1 - \lp \frac{n-1}{n}\rp^K \rp + B \mbb E (\delta_T | \delta_{t_0} = 0 ).
    \]
\end{lemma}
\begin{proof}
    Our proof closely follows that of \cite{Hardt2016}, just with a small distinction. After obtaining the inequality
    \[
        \mbb E | f(\bar{\vw}^T; \Xi) - f(\bar{\vv}^T; \Xi)| \le B \mbb E (\delta_T | \delta_{t_0} = 0) + \mbb P\{\mcl E^c\} (\sup f),
    \]
    where the event $\mcl E$ denotes the event that $\delta_{t_0} = 0$, we similarly need to bound $P\{\mcl E^c\}$. Defining the random variable $I$ to assume the index of the first time step in which DFedAvg uses the example $\xi^\ast$, which occurs at node $i^\ast \in [N]$ and is located in the $j^\ast \in [n]$ entry of $\tilde{\mcl D}_{i^\ast}$. We have
    \[
        P\{\mcl E^c\} = \mbb P \{\delta{t_0} \not= 0\} \le \mbb P\{I \le t_0\} \le \sum_{t=1}^{t_0} \mbb P\{I = t\}.
    \]
    Now since the draws at each round $t$ of DFedAvgM are sampled uniformly at random across both the nodes and the local datasets (that is $\xi^{t,k}_i \sim Unif(D_i)$ with replacement across iterations $k$), then we have that 
    \[
        \mbb P\{I = t\} = 1 - \mbb P\{I \not= t\} = 1 - \lp\frac{n-1}{n}\rp^K,
    \]  
    from which we conclude the proof.
\end{proof}

\begin{lemma} \label{lemma:bd-stepsize}

If $\eta_t \le \frac{c}{t}$ for $t =1, 2, \ldots$, then
\[
    \sum_{j=1}^t \eta_{t-j}\lambda^j \le \frac{C_\lambda}{t} 
\]
where 
\[
    C_\lambda := \min\left\{ 2\lambda, \frac{1}{\ln \frac{1}{\lambda}}\lambda^{\frac{1}{\ln \frac{1}{\lambda}} } \right\} +  \min\left\{ 4\lambda \ln \frac{1}{\lambda}, \frac{4}{\ln \frac{1}{\lambda}} \lambda^{\frac{2}{\ln \frac{1}{\lambda}}}\right\} + \min\left\{ 2\lambda, \frac{1}{\ln \frac{1}{\lambda}} \lambda^{\frac{1}{\ln \frac{1}{\lambda}}}  \right\} + \frac{2}{\ln \frac{1}{\lambda}}.
\]
\end{lemma}
\begin{proof}
\begin{align}
    &\quad \sum_{j=1}^t \eta_{t-j}\lambda^j\\
    &= \sum_{j=1}^t \eta_j \lambda^{t-j} \le \lambda^t \sum_{j=1}^t \frac{\lambda^{-j}}{j} \nonumber \\
    &= \lambda^{t} + \lambda^t \sum_{j=2}^{t} \frac{\lambda^{-j}}{j} \nonumber\\
    &\le \lambda^{t} + \lambda^t \sum_{j=2}^{t} \int_{j-1}^j \frac{\lambda^{-x}}{x} dx \\
    &= \lambda^{t} + \lambda^t \int_{1}^{t} \frac{\lambda^{-x}}{x} dx = \lambda^{t-1} + \lambda^t \lp \int_{1}^{t/2} \frac{\lambda^{-x}}{x} dx + \int_{t/2}^t \frac{\exp(x \ln (\frac{1}{\lambda}))}{x} dx \rp \label{eq:lambda-int} , 
\end{align}
which with $\lambda < 1$ we have that $\ln (1/\lambda) > 0$ and so we can simplify the integrals as
\[
    \int_{1}^{t/2} \frac{\lambda^{-x}}{x} dx \le \frac{2}{t} \int_1^{t/2} \lambda^{-x} dx  = \frac{2}{t \ln \frac{1}{\lambda}}\lp \lambda^{-t/2} - \lambda^{-t}\rp \le  \frac{2\lambda^{-t/2} }{t \ln \frac{1}{\lambda}}
\]
and
\begin{align*}
    &\quad \int_{t/2}^t \frac{\exp(x \ln (\frac{1}{\lambda}))}{x} dx\\
    &= \int_{1}^{t/2} \frac{1}{x} \sum_{k=0}^{\infty}  \frac{(\ln \frac{1}{\lambda})^k x^k}{k!} dx \\
    &= \int_{1}^{t/2} \frac{1}{x} dx +  (\ln \frac{1}{\lambda})\int_{1}^{t/2} 1 dx +  \sum_{k=2}^{\infty}  \frac{(\ln \frac{1}{\lambda})^k}{k!} \int_1^{t/2} x^{k-1} dx \\
    &= 1 - \frac{4}{t^2} + \frac{1}{2}(t - 2)\ln \frac{1}{\lambda} + \sum_{k=2}^{\infty}  \frac{(\ln \frac{1}{\lambda})^k}{(k)k!} \lp \lp\frac{t}{2}\rp^{k} - 1\rp \\
    &\le 1 - \frac{4}{t^2} + \frac{1}{2}(t - 2)\ln \frac{1}{\lambda} + \frac{1}{2} \sum_{k=2}^{\infty}  \frac{(\ln \frac{1}{\lambda})^k}{k!} \lp \lp\frac{t}{2}\rp^{k} - 1\rp \\
    &= 1 - \frac{4}{t^2} + \frac{1}{2}(t - 2)\ln \frac{1}{\lambda} + \frac{1}{2} \lp \frac{1}{\lambda^{t/2}} - 1 + \frac{t}{2}\ln \frac{1}{\lambda} - \left[\frac{1}{\lambda} - 1 + \ln \frac{1}{\lambda}  \right]   \rp \\
    &= 1 - \frac{4}{t^2} + \frac{1}{2}(t - 2)\ln \frac{1}{\lambda} + \frac{1}{2} \lp \frac{1}{\lambda^{t/2}} - \frac{1}{\lambda} + \frac{1}{2}(t-2)\ln \frac{1}{\lambda}  \rp \\
    &= 1 - \frac{4}{t^2} + \frac{3}{4}(t - 2)\ln \frac{1}{\lambda} + \frac{1}{2\lambda^{t/2}} - \frac{1}{2\lambda}.
\end{align*}
Plugging this result into (\ref{eq:lambda-int}), we obtain 
\begin{align} \label{eq:lambda-pre-bound}
    \sum_{j=1}^t \eta_{t-j}\lambda^j &\le \lambda^{t-1} + \lambda^t \lp 1 - \frac{4}{t^2} + \frac{3}{4}(t - 2)\ln \frac{1}{\lambda} + \frac{1}{2\lambda^{t/2}} - \frac{1}{2\lambda} + \frac{2\lambda^{-t/2}}{t\ln\frac{1}{\lambda}} \rp \nonumber\\
    &= \frac{\lambda^{t-1}}{2} + \lambda^t\lp 1 - \frac{4}{t^2} + \frac{3}{4}(t - 2)\ln \frac{1}{\lambda}\rp + \lambda^{t/2}\lp \frac{1}{2} + \frac{2}{t\ln\frac{1}{\lambda}}\rp \nonumber \\
    &\le \frac{1}{t}\left[\lambda^{t} \lp t + \ln \frac{1}{\lambda} t^2 \rp + \lambda^{t/2} \lp t + \frac{2}{\ln \frac{1}{\lambda}} \rp \right],
\end{align}
where in the last line we have used that $t-1 \ge t/2$ for $t \ge 2$.

Seeking a uniform bound over $t = 2, 3, \ldots$, we bound each of the last two terms of the right hand side of the above equation. It is easy to check that
\begin{align*}
    t\lambda^{t} &\le \min\left\{ 2\lambda^2, \frac{1}{\ln \frac{1}{\lambda}}\lambda^{\frac{1}{\ln \frac{1}{\lambda}} } \right\},\\
    t^2 \lambda^{t} &\le \min\left\{ 4\lambda^2, \frac{4}{\lp \ln \frac{1}{\lambda}\rp^2} \lambda^{\frac{2}{\ln \frac{1}{\lambda}}}\right\}, \\
    t \lambda^{t/2} &\le \min\left\{ 2\lambda, \frac{1}{\ln \frac{1}{\lambda}} \lambda^{\frac{1}{\ln \frac{1}{\lambda}}}  \right\},
\end{align*}
where we have noted that each of these functions are decreasing functions of $t$. Therefore, our bound becomes:
\begin{align*}
    &\quad \sum_{j=1}^t \eta_{t-j}\lambda^j\\
    &\le \frac{1}{t}\left[ \min\left\{ 2\lambda^2, \frac{1}{\ln \frac{1}{\lambda}}\lambda^{\frac{1}{\ln \frac{1}{\lambda}} } \right\} +  \min\left\{ 4\lambda^2 \ln \frac{1}{\lambda}, \frac{4}{\ln \frac{1}{\lambda}} \lambda^{\frac{2}{\ln \frac{1}{\lambda}}}\right\} + \min\left\{ 2\lambda, \frac{1}{\ln \frac{1}{\lambda}} \lambda^{\frac{1}{\ln \frac{1}{\lambda}}}  \right\}\right]\\
    &\qquad \qquad + \frac{2}{t\ln \frac{1}{\lambda}}  \\
    &= \frac{1}{t}\left[ 2\lambda^2 + 4\lambda^2 \ln \frac{1}{\lambda} + \frac{2}{\ln \frac{1}{\lambda}} +  \min\left\{ 2\lambda, \frac{1}{\ln \frac{1}{\lambda}} \lambda^{\frac{1}{\ln \frac{1}{\lambda}}}  \right\} \right] \\
    &=: \frac{C_\lambda}{t}. 
\end{align*}

We note that all terms of $C_\lambda$ except for $2/\ln \frac{1}{\lambda}$ are {\it uniformly bounded} on $\lambda \in (0,1)$. It is true that $2/\ln \frac{1}{\lambda} \rightarrow \infty$ as $\lambda \rightarrow 1^-$, but for each $\lambda < 1$ this bound $C_\lambda$ is valid. 

\end{proof}

\section*{Acknowledgments}~
This material is based on research sponsored by the NSF grant DMS-1924935 and DMS-1952339, and the DOE grant DE-SC0021142.


\end{document}